\documentclass[12pt]{article}
\usepackage[margin=1in]{geometry}
\usepackage{amsmath}
\usepackage{amsthm}
\usepackage{authblk}
\usepackage{adjustbox}
\usepackage{graphicx}
\usepackage{booktabs}
\usepackage{blkarray}
\usepackage[utf8]{inputenc}
\usepackage{csquotes}
\usepackage[backend=biber,style=apa,citestyle=authoryear]{biblatex}
\usepackage{amssymb}
\usepackage{float}
\usepackage{commath} %for norm
\usepackage{mathrsfs}
\usepackage{tikz-cd}
\usepackage{algorithm}
\usepackage{algpseudocode}

\usepackage[english]{babel}

\newtheorem{theorem}{Theorem}[section]

\newtheorem{prop}{Proposition}[section]

\newtheorem{corollary}{Corollary}[theorem]

\title{A stochastic correlation extension of the Vasicek credit risk model}
\author[1]{Dhruv Bansal\thanks{dbban3961@gmail.com}}
\author[2]{Mayank Goud\thanks{mayankgoud344@gmail.com}}
\author[3]{Sourav Majumdar\thanks{souravm@iitk.ac.in}}
\affil[1]{Department of Aerospace Engineering, Indian Institute of Technology Kanpur, India}
\affil[2]{Department of Economic Sciences, Indian Institute of Technology Kanpur, India}
\affil[3]{Department of Management Sciences, Indian Institute of Technology Kanpur, India}
\date{}

\begin{document}

\maketitle

\begin{abstract}
In this paper we extend the Vasicek credit risk model by modelling the correlation as a continuous-time process. The models for correlation follow diffusion processes on the circle. In particular we work with the Circular Brownian motion and a mean-reverting von Mises process. The analytical and computational tractability of these circular diffusion enables us to derive terminal asset and loss distributions by averaging their conditional laws over the law of time-averaged correlation, and evaluate path-dependent probabilities by Monte Carlo simulation. Simulations distinguish terminal joint default, joint survival, first-to-default, and joint first passage, and show how stochastic correlation reallocates probability between concordant and discordant outcomes. We fit both specifications to U.S. bank charge-off data, demonstrating that the framework remains tractable for probabilistic analysis and statistical estimation. \end{abstract}

\section{Introduction}\label{sec:introduction}

Credit portfolio risk is fundamentally a problem of dependence. Even when marginal default probabilities are well estimated, portfolio tail risk is determined by the clustering of defaults under common economic stress. Two principal approaches that are used to model default risk (\cite{bakshi2022decoding}) are structural models and reduced-form, or intensity-based models. The structural approach originates with \cite{merton1974pricing}. For a portfolio of obligors, \cite{vasicek2002distribution} extends the Merton framework by conditioning default probabilities on a single systematic factor. In the large homogeneous portfolio limit, the resulting default rate has a closed-form distribution, the Vasicek distribution, which permits tractable calculations of portfolio VaR, regulatory capital, and tranche loss probabilities \parencite{vasicek2002distribution}.

This tractability has made the Vasicek limit central to both industry practice and regulation. \cite{gordy2003risk} shows that portfolio-invariant capital charges arise naturally in an asymptotic single-factor model, and the Basel II Internal Ratings-Based (IRB) risk-weight functions are explicitly motivated and explained by this asymptotic single risk factor (ASRF) construction (\cite{BCBS2005IRB}). Portfolio invariance requires that the capital charge for an exposure depend only on that exposure's risk characteristics, rather than on the composition of the remaining portfolio (\cite{gordy2003risk}). Under a single systematic factor, this property emerges in the limit of a large, well-diversified portfolio and leads directly to the Vasicek distribution. The asset correlation parameter then determines the extent of default clustering under stress and, consequently, the relevant tail loss quantiles and regulatory capital charges.

Supervisory calibration therefore depends critically on correlation. The EBA staff paper on the IRB supervisory formula emphasises that asset correlation enters the supervisory risk-weight function directly and affects the transformation of expected probability of default (PD) into stressed, or worst-case, default rates (\cite{CasellinaSalisTessioreUgoccioniVaretto2023IRB}). Since the formula combines PD with asset correlation to produce a stressed default rate, even modest dependence misspecification can materially alter capital requirements.

The classical Vasicek model nevertheless fixes the asset correlation parameter $\rho$ over time. This restriction is difficult to reconcile with the empirical and economic state dependence of default clustering which is that correlations tend to rise under stress, and uncertainty about dependence is itself a source of correlation risk. This state dependence bears directly on the long-recognised cyclicality of Basel-style capital. Simulations in \cite{GoodhartSegoviano2004BaselProcyclicality} show that the IRB approach can generate lower capital requirements during expansions and markedly higher requirements during recessions, with substantially greater cyclical variation than alternative approaches. They conclude that ``procyclicality may well still be a serious problem with Basel II''. Similarly, \cite{CorcosteguiGonzalezMosqueraMarceloTrucharte2003ProcyclicalCapital} argue that the mechanical response of internal ratings-based capital to the business cycle can lower required capital in favourable conditions and raise it in unfavourable conditions, potentially inducing banks to tighten credit during recessions and amplify downturns. Using a rating system estimated over a business cycle, they quantify how macroeconomically driven obligor-grade migrations translate into changes in capital under the Basel proposals.

Even when the ASRF structure is retained, correlation calibration remains a material source of model risk. \cite{TarashevZhu2008ASRFErrors} note that portfolio invariance in the IRB framework rests on the single-factor and perfect-granularity assumptions. Violations of these assumptions may be relatively minor for large portfolios, but plausible calibration errors in both the level and dispersion of asset-return correlations can produce economically significant errors in portfolio credit-risk measures and capital. A fixed correlation input therefore omits an important source of risk.

The instability of a single correlation input is also visible in empirical estimates. \cite{DuellmannScheicherSchmieder2007AssetCorrelation} document substantial variation in estimated asset correlations across portfolios and show that averaging correlation parameters can materially understate portfolio VaR. In the Basel II setting, \cite{Lopez2002AssetCorrelationPDSize} emphasises that the ASRF framework is integral to regulatory capital and that average asset correlation is the key regulatory parameter linking obligor-level PDs to portfolio tail risk. More recently, \cite{LeeLinYang2011ProcyclicalCorrelation} find asymmetric and procyclical behaviour in asset correlations: correlations rise in downturns, are negatively related to default probability, and are positively related to firm size. A constant-correlation specification may therefore suppress systematic dependence variation precisely when tail clustering matters most.

Previous work in the credit-risk literature have considered aspects that have been allowed to evolve over time. \textcite{taufer2007modelling} models aggregate default rates by a mean-reverting diffusion with a given marginal distribution and autocorrelation structure. \textcite{czado2008modeling} study dependence in default probabilities across both time and rating categories by incorporating macroeconomic covariates. In a structural setting, \textcite{zhang2009first} develop Monte Carlo methods for multivariate first-passage problems and correlated defaults under jump-diffusion asset dynamics.  \textcite{miao2019extending} introduce common credit events with persistent effects to strengthen joint-default dependence and price multi-name products. \textcite{hocht2010pricing} combine stochastic recovery and default intensity in a tractable framework for credit derivatives. 

The limitations of a single correlation parameter are not confined to regulatory credit models. They also arise in financial models built on copula dependence, where effective correlation may change across states, maturities, and market regimes. This has led to stochastic and local correlation specifications. In credit derivatives, random factor loadings and stochastic or local correlation have been used to reproduce tranche and skew behaviour and to capture stronger default clustering under stress; see, e.g., \cite{andersen2004extensions,burtschell2007beyond,hull2005valuation}. From a risk-measurement perspective, random correlation replaces a single Vasicek loss distribution with a mixture over correlation states. Such mixtures typically thicken the tails and make extreme quantiles more sensitive to dependence.

Existing stochastic correlation models draw on several constructions. These include bounded functionals of Brownian motions (\cite{teng2016modelling,teng2016heston}), Wishart processes (\cite{asai2009structure,fonseca2007option,da2014estimating,da2011riding,da2011hedging,gnoatto2014explicit,grasselli2016flexible,da2015analytic}), and Jacobi processes (\cite{ma2009pricing}). However, many of them do not have transition densities available in closed form

For regulatory applications any extension must also remain operationally tractable. In the ASRF framework, idiosyncratic risk is diversified away and a single systematic factor drives dependence, which allows capital charges to be expressed analytically. As \cite{GordyLuetkebohmert2013GranularityAdjustment} stress the importance of having simple closed-form capital rule because of the transparencey they give across institutions. Yet concentration risk and correlation structure are precisely the dimensions along which actual portfolios depart from the ASRF benchmark. The Basel Committee's Research Task Force observes that name concentration violates perfect granularity, whereas sector concentration may require more than one systematic factor \cite{BCBS2006ConcentrationWP15}. The same report identifies stable and reliable estimation of asset correlations across exposures as a central practical difficulty. These constraints call for a dependence extension that remains analytically transparent and computationally efficient.

We impose the ASRF correlation bounds through a circular latent state. For a large homogeneous one-factor portfolio, nonnegative equicorrelation is required for positive semidefiniteness at arbitrary portfolio size. We therefore set
\[
R_t = \cos^2(\varphi_t), \qquad \varphi_t \in \mathbb{S}^1,
\]
where \(\varphi_t\) is a diffusion on the circle. The parametrization \(R_t=\cos^2(\varphi_t)\) also reflects the
correlation restrictions of the ASRF limit. In a homogeneous
\(n\)-obligor portfolio, a common pairwise correlation \(\rho\) is
admissible only if
\[
-\frac{1}{n-1}\leq \rho\leq 1.
\]
The lower bound converges to zero as \(n\to\infty\). Consequently, a
common correlation parameter that remains admissible for the
asymptotically fine-grained portfolios considered by
\textcite{gordy2003risk} must lie in \([0,1]\). The transformation
\(R_t=\cos^2(\varphi_t)\) enforces this restriction pathwise.
This construction accommodates both Circular Brownian motion and the mean-reverting von Mises process \cite{majumdar2024diffusion}. It also separates the credit quantities determined by time-averaged correlation from those that depend on the entire correlation path. We characterize terminal asset and loss distributions through time-averaged correlation, evaluate barrier events by full-path simulation, and measure the resulting departures from constant-correlation Vasicek benchmarks.

Section~\ref{sec:model} presents the model, expresses terminal asset and loss distributions in terms of time-averaged correlation, and develops additive-functional approximations. Section~\ref{sec:barrier-survival} gives full-path Monte Carlo estimators for joint survival, first-to-default, and joint first passage. Section~\ref{sec:simulation} reports benchmark and controlled parameter experiments. Section~\ref{sec:empirical} fits Circular Brownian motion and the von Mises process to U.S.\ bank charge-off rates by particle-filter simulated maximum likelihood under the exact Vasicek loss density, and maps the two filtered dependence distributions into conditional two-year structural probabilities. Section~\ref{sec:conclusion} concludes with limitations and directions for future work.\section{Model}\label{sec:model}

Let $S_{i,t}$ be the asset value of obligor $i$ and,
\[
X_{i,t}=\log S_{i,t}, \qquad x_{i,0}=\log S_{i,0}.
\]
We work with the following log-asset dynamics, written directly in one-factor ASRF form.  
\begin{equation}\label{eq:corrected-log-asset-dynamics}
dX_{i,t}
=
m_i\,dt
+
\sigma_i\left(\sqrt{R_t}\,dB_{0,t}+\sqrt{1-R_t}\,dB_{i,t}\right),
\qquad i=1,\ldots,n,
\end{equation}
where
\[
m_i=\mu_i-\frac12\sigma_i^2,
\]
$B_{0,t},B_{1,t},\ldots,B_{n,t}$ are independent Brownian motions, and $R_t\in[0,1]$ is the instantaneous ASRF asset correlation process.  Thus $B_{0,t}$ is the systematic Brownian factor and $B_{i,t}$ is the idiosyncratic Brownian factor of obligor $i$.

The dependence process is generated by a circular diffusion.  Let $\varphi_t$ be a diffusion on the circle $\mathbb S^1$, independent of $B_{0,t},B_{1,t},\ldots,B_{n,t}$, with generator
\begin{equation}\label{eq:angle-generator}
\mathcal A_\varphi f(\varphi)=a(\varphi)f'(\varphi)+\frac12\eta^2(\varphi)f''(\varphi),
\end{equation}
with periodic boundary conditions.  We set
\begin{equation}\label{eq:R-cos2}
R_t=R(\varphi_t)=\cos^2(\varphi_t),
\end{equation}
Two useful processes to model $\varphi_t$ are, see \cite{majumdar2024diffusion},
\begin{enumerate}
\item Circular Brownian motion
\[
d\varphi_t=\sigma_\varphi dB_t
\]
\item the von Mises process
\[
d\varphi_t=-\lambda\sin(\varphi_t-\mu_\varphi)dt+
\sigma_\varphi dB_t.
\]
\end{enumerate}
Here $B_t$ is a Brownian motion independent of the asset Brownian motions.  The role of the circular diffusion is to provide a bounded and interpretable state variable for dependence.  In the ASRF setting we use the nonnegative mapping $R_t=\cos^2(\varphi_t)$ rather than $\rho_t=\cos(\theta_t)$, because the large homogeneous one-factor portfolio model requires nonnegative equicorrelation in the asymptotic credit-portfolio interpretation.

\begin{prop}\label{prop:psd-correlation}
Consider a finite portfolio of \(n\ge2\) obligors with log-asset processes
\(X_1,\ldots,X_n\) satisfying \eqref{eq:corrected-log-asset-dynamics}.
Assume that \(\sigma_i>0\) for each \(i\). Let
\(A_t^{(n)}=(a_{ij}(t))_{1\le i,j\le n}\) denote the instantaneous
covariance matrix, defined by $d\langle X_i,X_j\rangle_t=a_{ij}(t)dt$.
Under the specification \(R_t=\cos^2(\varphi_t)\),
\[
a_{ij}(t)=
\begin{cases}
\sigma_i^2, & i=j\\
\sigma_i\sigma_jR_t, & i\neq j.
\end{cases}
\]
Consequently, the normalized instantaneous correlation matrix
\(C_t^{(n)}=(c_{ij}(t))_{1\le i,j\le n}\), where
\[
c_{ij}(t)=\frac{a_{ij}(t)}{\sqrt{a_{ii}(t)a_{jj}(t)}},
\]
has entries
\[
c_{ij}(t)=
\begin{cases}
1, & i=j,\\
R_t, & i\neq j.
\end{cases}
\]
Since \(R_t\in[0,1]\), \(C_t^{(n)}\) is positive semidefinite for every
\(n\ge2\).
\end{prop}

\begin{proof}
Since \(R_t=\cos^2(\varphi_t)\), we have \(R_t\in[0,1]\). From
\eqref{eq:corrected-log-asset-dynamics},
\[
dX_{i,t}
=
m_i\,dt+\sigma_i\sqrt{R_t}\,dB_{0,t}
+\sigma_i\sqrt{1-R_t}\,dB_{i,t}.
\]
Using the independence of \(B_0,B_1,\ldots,B_n\), for \(i\neq j\),
\[
d\langle X_i,X_j\rangle_t
=
\sigma_i\sigma_jR_t\,dt,
\]
while
\[
d\langle X_i,X_i\rangle_t
=
\sigma_i^2R_t\,dt+\sigma_i^2(1-R_t)\,dt
=
\sigma_i^2\,dt.
\]
Thus
\[
a_{ij}(t)=
\begin{cases}
\sigma_i^2, & i=j,\\
\sigma_i\sigma_jR_t, & i\neq j.
\end{cases}
\]
After normalization,
\[
c_{ij}(t)
=
\frac{a_{ij}(t)}{\sqrt{a_{ii}(t)a_{jj}(t)}}
=
\begin{cases}
1, & i=j,\\
R_t, & i\neq j.
\end{cases}
\]

It remains to verify positive semidefiniteness of \(C_t^{(n)}\).  Let
\(x=(x_1,\ldots,x_n)^\top\in\mathbb R^n\).  Then
\[
\begin{aligned}
x^\top C_t^{(n)}x
&=
\sum_{i=1}^n x_i^2
+
2R_t\sum_{1\le i<j\le n}x_ix_j \\
&=
(1-R_t)\sum_{i=1}^n x_i^2
+
R_t\left(\sum_{i=1}^n x_i\right)^2 .
\end{aligned}
\]
Both terms on the right-hand side are nonnegative because \(R_t\in[0,1]\).
Hence
\[
x^\top C_t^{(n)}x\ge0
\qquad\text{for all }x\in\mathbb R^n.
\]
Therefore,
\(C_t^{(n)}\) is positive semidefinite.
\end{proof}

\subsection{Vasicek loss density under stochastic correlation}
\label{subsec:corrected-vasicek-density}

We first isolate the Vasicek loss density kernel. This kernel will later be
used as the observation density for estimating effective dependence from
portfolio loss data.

\begin{prop}\label{prop:vasicek-kernel}
Let \(p\in(0,1)\) denote the unconditional default probability, and $q:=\Phi^{-1}(p)$,
where \(\Phi\) is the standard normal distribution function. Let
\(\rho\in(0,1)\) denote the asset-correlation parameter. Let \(Y\) be the
systematic standard normal factor and let
\(\varepsilon_1,\ldots,\varepsilon_n\) be idiosyncratic standard normal
variables, independent of \(Y\) and of each other. The default
indicator of obligor \(i\) is defined by
\[
D_i(\rho)
=
\mathbf 1\{\sqrt \rho Y+\sqrt{1-\rho}\varepsilon_i\le q\},
\qquad i=1,\ldots,n.
\]
Let
\[
L_n(\rho)=\sum_{i=1}^n w_{i,n}D_i(\rho),
\qquad
\sum_{i=1}^n w_{i,n}=1,
\qquad
\sum_{i=1}^n w_{i,n}^2\to0.
\]
Then
\[
L_n(\rho)\to
L(\rho)
=
\Phi\left(\frac{q-\sqrt \rho Y}{\sqrt{1-\rho}}\right)
\]
in \(L^2\), and the limiting loss fraction has density
\begin{equation}\label{eq:vasicek-density-kernel}
\begin{aligned}
f(x;p,\rho)
&=
\sqrt{\frac{1-\rho}{\rho}}
\\
&\quad{}\times
\exp\left\{
-\frac{
\left(\sqrt{1-\rho}\Phi^{-1}(x)-\Phi^{-1}(p)\right)^2
}{2\rho}
+\frac12\left(\Phi^{-1}(x)\right)^2
\right\},
\qquad 0<x<1.
\end{aligned}
\end{equation}
Moreover, \(\mathbb E[L(\rho)]=p\).
\end{prop}
\begin{proof}
Conditional on \(Y=y\),
\[
\mathbb P(D_i(\rho)=1\mid Y=y)
=
\Phi\left(\frac{q-\sqrt \rho y}{\sqrt{1-\rho}}\right)
:=l(y,\rho).
\]
Hence
\[
\mathbb E[L_n(\rho)\mid Y]=l(Y,\rho),
\]
and, by conditional independence,
\[
\operatorname{Var}(L_n(\rho)\mid Y)
=
\sum_{i=1}^n w_{i,n}^2l(Y,\rho)(1-l(Y,\rho))
\le
\frac14\sum_{i=1}^n w_{i,n}^2.
\]
Therefore
\[
\mathbb E\left[(L_n(\rho)-l(Y,\rho))^2\right]
\le
\frac14\sum_{i=1}^n w_{i,n}^2
\to0,
\]
so \(L_n(\rho)\to L(\rho)=l(Y,\rho)\) in \(L^2\).

Let \(z_x=\Phi^{-1}(x)\). Since
\[
L(\rho)\le x
\quad\Longleftrightarrow\quad
Y\ge
\frac{q-\sqrt{1-\rho}\,z_x}{\sqrt \rho},
\]
we have
\[
F_{L\mid \rho}(x)
=
\Phi\left(
\frac{\sqrt{1-\rho}\,\Phi^{-1}(x)-q}{\sqrt \rho}
\right).
\]
After differentiating we get, with $\phi$ denoting the standard normal density,
\[
f_{L\mid \rho}(x)
=
\sqrt{\frac{1-\rho}{\rho}}
\frac{
\phi\left((\sqrt{1-\rho}\,\Phi^{-1}(x)-q)/\sqrt \rho\right)
}{
\phi(\Phi^{-1}(x))
},
\]
which is \eqref{eq:vasicek-density-kernel} after substituting
\(q=\Phi^{-1}(p)\). Finally,
\[
\mathbb E[L(\rho)]
=
\mathbb E[\mathbb P(D_i(\rho)=1\mid Y)]
=\mathbb E[\mathbb E[D_i(\rho)\mid Y]] 
=\mathbb E[D_i(\rho)]
=
\mathbb P(D_i(\rho)=1)
=
p.
\]
\end{proof}

\begin{prop}[Vasicek density with stochastic correlation]
\label{prop:vasicek-density-stochastic-correlation}
Let the log price follow
\eqref{eq:corrected-log-asset-dynamics}, with $R_t=\cos^2(\varphi_t)$.
For a fixed horizon \(T>0\), define the time-averaged correlation
\[
\overline R_T=\frac1T\int_0^T R_s\,ds.
\]
Conditional on the correlation path \(R_{[0,T]}\), the limiting portfolio
loss fraction at horizon \(T\) has density
\begin{equation}\label{eq:conditional-vasicek-average}
f_{L_T\mid R_{[0,T]}}(x\mid p,R_{[0,T]})
=
f(x;p,\overline R_T),
\qquad 0<x<1,
\end{equation}
where \(f\) is the limiting loss fraction density in
\eqref{eq:vasicek-density-kernel}. Consequently, the unconditional loss
density is
\begin{equation}\label{eq:unconditional-vasicek-mixture}
f_{L_T}(x\mid p)
=
\mathbb E\left[f(x;p,\overline R_T)\right].
\end{equation}
If \(\overline R_T\) has density \(h_T\) on \((0,1)\), then
\begin{equation}\label{eq:vasicek-density-mixture-integral}
f_{L_T}(x\mid p)
=
\int_0^1 f(x;p,\rho)h_T(\rho)d\rho.
\end{equation}
\end{prop}

\begin{proof}
Conditional on the correlation path \(R_{[0,T]}\), integration of
\eqref{eq:corrected-log-asset-dynamics} gives
\[
\frac{X_{i,T}-X_{i,0}-m_iT}{\sigma_i\sqrt T}
=
\frac1{\sqrt T}\int_0^T\sqrt{R_s}\,dB_{0,s}
+
\frac1{\sqrt T}\int_0^T\sqrt{1-R_s}\,dB_{i,s}.
\]
The first term is common across obligors. By Itô isometry,
\[
\operatorname{Var}\left(
\frac1{\sqrt T}\int_0^T\sqrt{R_s}\,dB_{0,s}
\,\middle|\,R_{[0,T]}
\right)
=
\frac1T\int_0^T R_s\,ds
=
\overline R_T.
\]
Similarly,
\[
\operatorname{Var}\left(
\frac1{\sqrt T}\int_0^T\sqrt{1-R_s}\,dB_{i,s}
\,\middle|\,R_{[0,T]}
\right)
=
\frac1T\int_0^T(1-R_s)\,ds
=
1-\overline R_T.
\]
The two terms are conditionally independent, since they are driven by
independent Brownian motions.

Hence, conditional on \(R_{[0,T]}\), the normalized log-asset return has the
same distributional form as the latent variable in
Proposition~\ref{prop:vasicek-kernel}, with the correlation parameter
\(\rho=\overline R_T\). Applying Proposition~\ref{prop:vasicek-kernel} gives
\[
f_{L_T\mid R_{[0,T]}}(x\mid p,R_{[0,T]})
=
f(x;p,\overline R_T),
\qquad 0<x<1.
\]

Taking expectations over the correlation path gives
\[
f_{L_T}(x\mid p)
=
\mathbb E\left[f(x;p,\overline R_T)\right].
\]
If \(\overline R_T\) has density \(h_T\) on \((0,1)\), then conditioning on
\(\overline R_T\) yields
\[
f_{L_T}(x\mid p)
=
\int_0^1 f(x;p,\rho)h_T(\rho)\,d\rho.
\]
\end{proof}

\subsection{Terminal joint default at horizon}\label{subsec:terminal-joint}

The next result replaces the terminal-correlation calculation by the correct stochastic-integral calculation.  The proof uses standard It\^o-isometry and conditional-Gaussian arguments for stochastic integrals \cite{karatzas1991brownian,oksendal2003stochastic}.  The relevant quantity is the time-averaged correlation $\overline R_t=t^{-1}\int_0^tR_sds$, not the terminal instantaneous value $R_t$.

\begin{theorem}[Joint asset-value density under stochastic correlation]
\label{thm:joint-dist-corrected}
Let \(X_{i,t}=\log S_{i,t}\), \(i=1,2\), follow
\eqref{eq:corrected-log-asset-dynamics}, with
\(R_t=\cos^2(\varphi_t)\). Let \(\overline R_T\) be the time-averaged
correlation, and assume
that \(\overline R_T<1\) almost surely.

Conditional on the correlation path \(R_{[0,T]}\), the vector
\((X_{1,T},X_{2,T})^\top\) is bivariate normal with mean vector
\[
\begin{pmatrix}
x_{1,0}+m_1T\\
x_{2,0}+m_2T
\end{pmatrix}
\]
and covariance matrix
\begin{equation}\label{eq:conditional-covariance-corrected}
\Sigma_T(R)
=
T
\begin{pmatrix}
\sigma_1^2
&
\sigma_1\sigma_2\overline R_T\\
\sigma_1\sigma_2\overline R_T
&
\sigma_2^2
\end{pmatrix}.
\end{equation}
In particular,
\[
\operatorname{Corr}
\left(
X_{1,T},X_{2,T}\mid R_{[0,T]}
\right)
=
\overline R_T.
\]

Consequently, conditional on \(R_{[0,T]}\), the joint density of
\((S_{1,T},S_{2,T})\) is
\begin{equation}\label{eq:conditional-lognormal-density-corrected}
f_{S_{1,T},S_{2,T}\mid R_{[0,T]}}
(s_1,s_2\mid R_{[0,T]})
=
g_T(s_1,s_2;\overline R_T),
\qquad s_1,s_2>0,
\end{equation}
where
\begin{equation}\label{eq:bivariate-lognormal-density}
g_T(s_1,s_2;\rho)
=
\frac{1}{
2\pi\sigma_1\sigma_2T s_1s_2\sqrt{1-\rho^2}
}
\exp\left[
-\frac{z_1^2-2\rho z_1z_2+z_2^2}
{2(1-\rho^2)}
\right],
\qquad 0\le \rho<1,
\end{equation}
with
\[
z_i
=
\frac{\log(s_i/S_{i,0})-m_iT}{\sigma_i\sqrt T},
\qquad i=1,2.
\]

The unconditional joint density is
\begin{equation}\label{eq:unconditional-joint-lognormal-mixture}
f_{S_{1,T},S_{2,T}}(s_1,s_2)
=
\mathbb E\left[
g_T(s_1,s_2;\overline R_T)
\right].
\end{equation}
If \(\overline R_T\) has density \(h_T\) on \((0,1)\), then
\begin{equation}\label{eq:joint-lognormal-integral-average}
f_{S_{1,T},S_{2,T}}(s_1,s_2)
=
\int_0^1
g_T(s_1,s_2;\rho)h_T(\rho)\,d\rho.
\end{equation}
\end{theorem}

\begin{proof}
Integrating \eqref{eq:corrected-log-asset-dynamics} over \([0,T]\) gives
\[
X_{i,T}
=
x_{i,0}+m_iT
+
\sigma_i\int_0^T\sqrt{R_s}\,dB_{0,s}
+
\sigma_i\int_0^T\sqrt{1-R_s}\,dB_{i,s},
\qquad i=1,2.
\]
Conditional on \(R_{[0,T]}\), the stochastic integrals are jointly Gaussian
with mean zero. By the It\^o isometry and the independence of
\(B_{0,t},B_{1,t},B_{2,t}\),
\[
\begin{aligned}
\operatorname{Var}(X_{i,T}\mid R_{[0,T]})
&=
\sigma_i^2
\left[
\int_0^T R_s\,ds
+
\int_0^T(1-R_s)\,ds
\right] \\
&=
\sigma_i^2T,
\qquad i=1,2,
\end{aligned}
\]
and
\[
\begin{aligned}
\operatorname{Cov}(X_{1,T},X_{2,T}\mid R_{[0,T]})
&=
\sigma_1\sigma_2
\operatorname{Var}\left(
\int_0^T\sqrt{R_s}\,dB_{0,s}
\,\middle|\,R_{[0,T]}
\right)\\
&=
\sigma_1\sigma_2\int_0^T R_s\,ds\\
&=
\sigma_1\sigma_2T\overline R_T.
\end{aligned}
\]
Therefore, conditional on \(R_{[0,T]}\),
\((X_{1,T},X_{2,T})^\top\) is bivariate normal with mean vector
\[
\begin{pmatrix}
x_{1,0}+m_1T\\
x_{2,0}+m_2T
\end{pmatrix}
\]
and covariance matrix \(\Sigma_T(R)\) given in
\eqref{eq:conditional-covariance-corrected}. Moreover,
\[
\begin{aligned}
\operatorname{Corr}
\left(X_{1,T},X_{2,T}\mid R_{[0,T]}\right)
&=
\frac{\sigma_1\sigma_2T\overline R_T}
{(\sigma_1\sqrt T)(\sigma_2\sqrt T)}\\
&=
\overline R_T.
\end{aligned}
\]

Since \(S_{i,T}=e^{X_{i,T}}\),
\[
\left|
\frac{\partial(\log s_1,\log s_2)}
{\partial(s_1,s_2)}
\right|
=
\frac1{s_1s_2}.
\]
Hence
\[
\begin{aligned}
f_{S_{1,T},S_{2,T}\mid R_{[0,T]}}
(s_1,s_2\mid R_{[0,T]})
&=
\frac{
f_{X_{1,T},X_{2,T}\mid R_{[0,T]}}
(\log s_1,\log s_2\mid R_{[0,T]})
}{s_1s_2}\\
&=
g_T(s_1,s_2;\overline R_T),
\end{aligned}
\]
which proves \eqref{eq:conditional-lognormal-density-corrected}.

Finally,
\[
f_{S_{1,T},S_{2,T}}(s_1,s_2)
=
\mathbb E\left[
g_T(s_1,s_2;\overline R_T)
\right].
\]
If \(\overline R_T\) has density \(h_T\), then
\[
\mathbb E\left[
g_T(s_1,s_2;\overline R_T)
\right]
=
\int_0^1
g_T(s_1,s_2;\rho)h_T(\rho)\,d\rho.
\]
\end{proof}

\begin{corollary}[Joint default probability at horizon \(T\)]
\label{cor:joint-default-horizon-corrected}
Let \(B_1,B_2>0\) be the default barriers at horizon \(T\), and define
\[
d_i
=
\frac{\log(B_i/S_{i,0})-m_iT}{\sigma_i\sqrt T},
\qquad i=1,2.
\]
Conditional on the correlation path \(R_{[0,T]}\),
\begin{equation}\label{eq:conditional-joint-default-average}
\mathbb P
\left(
S_{1,T}\le B_1,\,
S_{2,T}\le B_2
\,\middle|\,
R_{[0,T]}
\right)
=
\Phi_2(d_1,d_2;\overline R_T),
\end{equation}
where \(\Phi_2(a,b;\rho)\) denotes the standard bivariate normal
distribution function with correlation \(\rho\). Consequently,
\begin{equation}\label{eq:joint-default-average-correlation}
\mathbb P
\left(
S_{1,T}\le B_1,\,
S_{2,T}\le B_2
\right)
=
\mathbb E\left[
\Phi_2(d_1,d_2;\overline R_T)
\right].
\end{equation}
If \(\overline R_T\) has density \(h_T\) on \((0,1)\), then
\begin{equation}\label{eq:joint-default-average-integral}
\mathbb P
\left(
S_{1,T}\le B_1,\,
S_{2,T}\le B_2
\right)
=
\int_0^1
\Phi_2(d_1,d_2;\rho)h_T(\rho)\,d\rho.
\end{equation}
\end{corollary}

\begin{proof}
For \(i=1,2\),
\[
\begin{aligned}
S_{i,T}\le B_i
&\Longleftrightarrow
X_{i,T}\le\log B_i\\
&\Longleftrightarrow
\frac{X_{i,T}-x_{i,0}-m_iT}{\sigma_i\sqrt T}
\le d_i.
\end{aligned}
\]
By Theorem~\ref{thm:joint-dist-corrected}, conditional on
\(R_{[0,T]}\), the standardized log-asset values are jointly standard
normal with correlation \(\overline R_T\). Hence
\[
\mathbb P
\left(
S_{1,T}\le B_1,\,
S_{2,T}\le B_2
\,\middle|\,
R_{[0,T]}
\right)
=
\Phi_2(d_1,d_2;\overline R_T).
\]
Taking expectations gives
\[
\mathbb P
\left(
S_{1,T}\le B_1,\,
S_{2,T}\le B_2
\right)
=
\mathbb E\left[
\Phi_2(d_1,d_2;\overline R_T)
\right].
\]
If \(\overline R_T\) has density \(h_T\), then
\[
\mathbb E\left[
\Phi_2(d_1,d_2;\overline R_T)
\right]
=
\int_0^1
\Phi_2(d_1,d_2;\rho)h_T(\rho)\,d\rho.
\]
\end{proof}

\section[Gaussian approximations for time averaged correlation]{Gaussian approximations for\\ time-averaged correlation}
\label{sec:time-averaged-correlation}

The terminal asset, loss, and joint-default formulae derived above must be averaged over the law of
\(\overline R_T\).  Although this law is generally not available in elementary
form, \(\overline R_T\) is an additive functional of the stationary circular
diffusion, so its mean and variance are determined by the stationary
autocovariance of \(R_t\).  Write
\[
m_R=\mathbb E[R_0],
\qquad
C_R(t)=\operatorname{Cov}(R_t,R_0).
\]
Stationarity gives
\begin{equation}\label{eq:time-average-stationary-variance}
\mathbb E[\overline R_T]=m_R,
\qquad
\operatorname{Var}(\overline R_T)
=
\frac{2}{T^2}\int_0^T(T-t)C_R(t)\,dt.
\end{equation}
Thus a Gaussian approximation to \(\overline R_T\) reduces to calculating, or
approximating, the covariance function \(C_R\).  The covariance expressions
below are exponential or finite sums of exponentials.  For later use, define
\begin{equation}\label{eq:time-average-exponential-integral}
\mathcal J(a,T)
=
\int_0^T(T-t)e^{-at}\,dt
=
\frac{T}{a}-\frac{1-e^{-aT}}{a^2},
\qquad a>0.
\end{equation}
Accordingly, a covariance component of the form \(ce^{-at}\) contributes
\(2c\mathcal J(a,T)/T^2\) to the variance of \(\overline R_T\).

\subsection{Circular Brownian motion}

\begin{prop}[Time-averaged correlation under Circular Brownian motion]
\label{prop:time-average-cbm-gaussian}
Suppose that \(\varphi_t\) follows the Circular Brownian motion specification in
Section~\ref{sec:model}, and initialize \(\varphi_0\) from the uniform
invariant distribution on \(\mathbb S^1\).  Then
\begin{equation}\label{eq:cbm-time-average-stationary-covariance}
\mathbb E[R_t]=\frac12,
\qquad
\operatorname{Cov}(R_t,R_0)
=
\frac18e^{-2\sigma_\varphi^2t}.
\end{equation}
Consequently,
\begin{equation}\label{eq:cbm-time-average-stationary-moments}
\mathbb E[\overline R_T]=\frac12,
\qquad
V_{T,\mathrm{stat}}^{\mathrm{CBM}}
:=
\operatorname{Var}(\overline R_T)
=
\frac{1}{8\sigma_\varphi^2T}
-
\frac{1-e^{-2\sigma_\varphi^2T}}
{16\sigma_\varphi^4T^2}.
\end{equation}
Moreover,
\begin{equation}\label{eq:cbm-time-average-clt}
\sqrt T\left(\overline R_T-\frac12\right)
\xrightarrow{D}
N\left(0,\frac{1}{8\sigma_\varphi^2}\right)
\qquad\text{as }T\to\infty.
\end{equation}
\end{prop}

\begin{proof}
Let
\[
f(\theta)=\cos(2\theta).
\]
The generator of Circular Brownian motion is
\[
\mathcal A_{\mathrm{CBM}}
=
\frac{\sigma_\varphi^2}{2}\frac{d^2}{d\theta^2},
\]
and therefore
\[
\mathcal A_{\mathrm{CBM}}f
=
-2\sigma_\varphi^2f.
\]
It follows from the semigroup eigenfunction relation that
\[
\mathbb E\left[f(\varphi_t)\mid\varphi_0\right]
=
e^{-2\sigma_\varphi^2t}f(\varphi_0).
\]
Under the uniform invariant distribution
\[
\pi_{\mathrm{unif}}(d\theta)=\frac{d\theta}{2\pi},
\]
we have
\[
\mathbb E_{\pi_{\mathrm{unif}}}[f(\varphi_0)]=0,
\qquad
\mathbb E_{\pi_{\mathrm{unif}}}[f(\varphi_0)^2]
=
\frac12.
\]
Hence, by the tower property,
\[
\begin{aligned}
\operatorname{Cov}
\left(f(\varphi_t),f(\varphi_0)\right)
&=
\mathbb E_{\pi_{\mathrm{unif}}}
\left[
\mathbb E\left[f(\varphi_t)\mid\varphi_0\right]
f(\varphi_0)
\right]\\
&=
e^{-2\sigma_\varphi^2t}
\mathbb E_{\pi_{\mathrm{unif}}}
\left[f(\varphi_0)^2\right]\\
&=
\frac12e^{-2\sigma_\varphi^2t}.
\end{aligned}
\]
Since
\[
R_t=\cos^2(\varphi_t)
=
\frac12\{1+f(\varphi_t)\},
\]
we obtain
\[
\operatorname{Cov}(R_t,R_0)
=
\frac18e^{-2\sigma_\varphi^2t},
\]
which proves
\eqref{eq:cbm-time-average-stationary-covariance}.

Applying
\eqref{eq:time-average-stationary-variance} and
\eqref{eq:time-average-exponential-integral} gives
\[
\begin{aligned}
\operatorname{Var}(\overline R_T)
&=
\frac{2}{T^2}
\int_0^T
(T-t)\frac18e^{-2\sigma_\varphi^2t}\,dt\\
&=
\frac{1}{4T^2}
\mathcal J(2\sigma_\varphi^2,T)\\
&=
\frac{1}{8\sigma_\varphi^2T}
-
\frac{1-e^{-2\sigma_\varphi^2T}}
{16\sigma_\varphi^4T^2}.
\end{aligned}
\]

For the limiting distribution, define the centred observable
\[
g(\theta)
=
R(\theta)-\frac12
=
\frac12\cos(2\theta).
\]
It satisfies
\[
\mathcal A_{\mathrm{CBM}}g
=
-2\sigma_\varphi^2g.
\]
Consequently, the Poisson equation
\[
\mathcal A_{\mathrm{CBM}}h=g
\]
admits the smooth periodic solution
\[
h(\theta)
=
-\frac{g(\theta)}{2\sigma_\varphi^2}.
\]
Thus \(g\) belongs to the domain of
\(\mathcal A_{\mathrm{CBM}}^{-1}\), and the additive functional
central limit theorem of
\cite{cattiaux2011central} (see Theorem 3.1 therein)
applies to
\[
\int_0^T g(\varphi_t)\,dt
=
T\left(\overline R_T-\frac12\right).
\]

Equivalently, setting
\[
\psi=-h
=
\frac{g}{2\sigma_\varphi^2},
\]
so that
\[
-\mathcal A_{\mathrm{CBM}}\psi=g,
\]
the asymptotic variance is
\[
\begin{aligned}
2\langle g,\psi\rangle_{L^2(\pi_{\mathrm{unif}})}
&=
\frac{1}{\sigma_\varphi^2}
\mathbb E_{\pi_{\mathrm{unif}}}
\left[g(\varphi_0)^2\right]\\
&=
\frac{1}{\sigma_\varphi^2}
\operatorname{Var}_{\pi_{\mathrm{unif}}}(R_0)\\
&=
\frac{1}{8\sigma_\varphi^2}.
\end{aligned}
\]
Therefore,
\[
\sqrt{T}
\left(
\overline R_T-\frac12
\right)
\overset{D}{\longrightarrow}
N\left(0,\frac{1}{8\sigma_\varphi^2}\right),
\]
which proves \eqref{eq:cbm-time-average-clt}.
\end{proof}
\subsection{Von Mises process}

For the von Mises process, let
\[
d\varphi_t
=
-\lambda\sin(\varphi_t-\mu_\varphi)dt
+\sigma_\varphi dB_t,
\qquad
\lambda>0,
\quad
\sigma_\varphi>0,
\]
and initialize the process from its invariant density
\begin{equation}\label{eq:vm-time-average-invariant}
\pi_{\mathrm{vM}}(\theta)
=
\frac{\exp(\kappa\cos(\theta-\mu_\varphi))}
{2\pi I_0(\kappa)},
\qquad
\kappa=\frac{2\lambda}{\sigma_\varphi^2}.
\end{equation}
Write
\begin{equation}\label{eq:vm-time-average-bessel-ratios}
\beta_j(\kappa)=\frac{I_j(\kappa)}{I_0(\kappa)},
\qquad j=2,4.
\end{equation}
Then the stationary mean of the correlation process is
\begin{equation}\label{eq:vm-time-average-stationary-mean}
m_{\mathrm{vM}}
=
\mathbb E_{\pi_{\mathrm{vM}}}[R_t]
=
\frac12\left[1+\beta_2(\kappa)\cos(2\mu_\varphi)\right].
\end{equation}

\begin{prop}[Gaussian limit for the time-averaged correlation]
\label{prop:time-average-vm-clt}
Suppose that the von Mises process is initialized from its invariant
distribution in \eqref{eq:vm-time-average-invariant}. Then
\begin{equation}
\label{eq:vm-time-average-clt}
\sqrt{T}
\left(
\overline R_T-m_{\mathrm{vM}}
\right)
\xrightarrow{D}
N\left(0,\tau_{\mathrm{vM}}^2\right),
\end{equation}
where
\begin{equation}
\label{eq:vm-time-average-asymptotic-variance}
\tau_{\mathrm{vM}}^2
=
2\int_0^\infty
\operatorname{Cov}_{\pi_{\mathrm{vM}}}
(R_t,R_0)\,dt.
\end{equation}
\end{prop}

\begin{proof}
Let
\[
X_t=(\varphi_t-\mu_\varphi)\pmod{2\pi}.
\]
Under the stationary initialization,
\[
X_t\sim\widetilde\pi_{\mathrm{vM}}(x)\,dx,
\qquad
\widetilde\pi_{\mathrm{vM}}(x)
=
\frac{e^{\kappa\cos x}}{2\pi I_0(\kappa)}.
\]
The standard Bessel moment identities give
\[
\mathbb E[\cos(2X_t)]=\beta_2(\kappa),
\qquad
\mathbb E[\sin(2X_t)]=0.
\]
Since
\[
R_t
=
\frac12
\left[
1+\cos(2\mu_\varphi)\cos(2X_t)
-\sin(2\mu_\varphi)\sin(2X_t)
\right],
\]
it follows that
\[
\mathbb E[R_t]
=
\frac12
\left[
1+\beta_2(\kappa)\cos(2\mu_\varphi)
\right]
=
m_{\mathrm{vM}}.
\]

Now define the smooth centred observable
\[
g(x)
=
\cos^2(x+\mu_\varphi)-m_{\mathrm{vM}}.
\]
The process \(X\) is a stationary, reversible, uniformly elliptic
diffusion on the compact circle. Its Markov semigroup therefore
converges exponentially on the mean-zero subspace of
\(L^2(\widetilde\pi_{\mathrm{vM}})\). In particular, \(g\) satisfies
the integrability condition in
\cite[Corollary~3.2]{cattiaux2011central}.
The functional central limit theorem of
\cite[Theorem~3.1]{cattiaux2011central} therefore leads to,
\[
\frac{1}{\sqrt T}
\int_0^T g(X_t)\,dt
\xrightarrow{D}
N(0,\tau_{\mathrm{vM}}^2).
\]
Since
\[
\frac{1}{\sqrt T}
\int_0^T g(X_t)\,dt
=
\sqrt T
\left(
\overline R_T-m_{\mathrm{vM}}
\right),
\]
this proves \eqref{eq:vm-time-average-clt}.

Finally, from stationarity it follows,
\[
\frac1T
\operatorname{Var}
\left(
\int_0^T g(X_t)\,dt
\right)
=
2\int_0^T
\left(1-\frac{t}{T}\right)
\operatorname{Cov}(R_t,R_0)\,dt.
\]
Since the covariance is integrable, from the dominated
convergence theorem it follows that,
\[
\tau_{\mathrm{vM}}^2
=
2\int_0^\infty
\operatorname{Cov}_{\pi_{\mathrm{vM}}}(R_t,R_0)\,dt.
\]
\end{proof}

To obtain an explicit finite time approximation, define
\[
f_c(x)
=
\cos(2x)-\beta_2(\kappa),
\qquad
f_s(x)
=
\sin(2x),
\]
and
\begin{equation}
\label{eq:vm-two-mode-variances}
V_c
=
\frac{1+\beta_4(\kappa)}{2}
-\beta_2(\kappa)^2,
\qquad
V_s
=
\frac{1-\beta_4(\kappa)}{2}.
\end{equation}
Let
\begin{equation}
\label{eq:vm-two-mode-rates}
\alpha_c
=
2\sigma_\varphi^2\frac{V_s}{V_c},
\qquad
\alpha_s
=
\sigma_\varphi^2
\frac{1+\beta_4(\kappa)}{V_s}.
\end{equation}

The centred correlation can be written as,
\[
R_t-m_{\mathrm{vM}}
=
\frac12\cos(2\mu_\varphi)f_c(X_t)
-
\frac12\sin(2\mu_\varphi)f_s(X_t).
\]
The invariant density is even, \(f_c\) is even, and \(f_s\) is
odd. Moreover, the von Mises process generator is invariant under the
reflection \(x\mapsto -x\), so its semigroup preserves parity.
Consequently,
\[
\operatorname{Cov}(f_c(X_t),f_s(X_0))
=
\operatorname{Cov}(f_s(X_t),f_c(X_0))
=
0
\]
for every \(t\geq0\). From routine calculations it follows that,
\[
\operatorname{Var}(f_c(X_0))=V_c,
\qquad
\operatorname{Var}(f_s(X_0))=V_s.
\]
For any centred smooth periodic function \(f\), stationarity and the
semigroup representation give
\[
\operatorname{Cov}(f(X_t),f(X_0))
=
\langle P_tf,f\rangle_{L^2(\widetilde\pi_{\mathrm{vM}})}.
\]
Therefore, by the definition of the infinitesimal generator and the
Dirichlet form identity for a reversible diffusion,
\[
\begin{aligned}
-\left.
\frac{d}{dt}
\operatorname{Cov}(f(X_t),f(X_0))
\right|_{t=0+}
&=
-\langle
\mathcal A_{\mathrm{vM}}f,f
\rangle_{L^2(\widetilde\pi_{\mathrm{vM}})} \\
&=
\frac{\sigma_\varphi^2}{2}
\mathbb E_{\widetilde\pi_{\mathrm{vM}}}
\left[f'(X_0)^2\right].
\end{aligned}
\]
See \cite[(1.4.1), (1.6.3), and
(1.11.1)-(1.11.2)]{bakry2014analysis}.

Define the two modal autocovariance functions by
\[
C_c(t)
=
\operatorname{Cov}\bigl(f_c(X_t),f_c(X_0)\bigr),
\qquad
C_s(t)
=
\operatorname{Cov}\bigl(f_s(X_t),f_s(X_0)\bigr).
\]
Since \(f_c\) and \(f_s\) are centred under the invariant
distribution, these may equivalently be written as
\[
C_c(t)
=
\mathbb E\bigl[f_c(X_t)f_c(X_0)\bigr],
\qquad
C_s(t)
=
\mathbb E\bigl[f_s(X_t)f_s(X_0)\bigr].
\]

Applying the covariance-derivative identity to \(f_c\) gives
\[
\begin{aligned}
-C_c'(0+)
&=
\frac{\sigma_\varphi^2}{2}
\mathbb E\!\left[f_c'(X_0)^2\right] \\
&=
\frac{\sigma_\varphi^2}{2}
\mathbb E\!\left[4\sin^2(2X_0)\right] \\
&=
2\sigma_\varphi^2V_s.
\end{aligned}
\]
Similarly,
\[
\begin{aligned}
-C_s'(0+)
&=
\frac{\sigma_\varphi^2}{2}
\mathbb E\!\left[f_s'(X_0)^2\right] \\
&=
\frac{\sigma_\varphi^2}{2}
\mathbb E\!\left[4\cos^2(2X_0)\right] \\
&=
\sigma_\varphi^2
\{1+\beta_4(\kappa)\}.
\end{aligned}
\]
We approximate each modal covariance by the single exponential that
matches its exact value and slope at the origin:
\[
C_c(t)\approx V_ce^{-\alpha_ct},
\qquad
C_s(t)\approx V_se^{-\alpha_st}.
\]
This gives
\begin{equation}
\label{eq:vm-two-mode-covariance}
C_R^{(2)}(t)
=
\frac14
\left[
\cos^2(2\mu_\varphi)V_ce^{-\alpha_ct}
+
\sin^2(2\mu_\varphi)V_se^{-\alpha_st}
\right].
\end{equation}

Substitution into the stationary time-average variance identity
yields
\begin{equation}
\label{eq:vm-two-mode-time-average-variance}
V_{T,(2)}^{\mathrm{vM}}
=
\frac{1}{2T^2}
\left[
\cos^2(2\mu_\varphi)V_c
\mathcal J(\alpha_c,T)
+
\sin^2(2\mu_\varphi)V_s
\mathcal J(\alpha_s,T)
\right].
\end{equation}
Motivated by Proposition~\ref{prop:time-average-vm-clt}, we use the
moment matched finite time approximation,
\begin{equation}
\label{eq:vm-two-mode-time-average-gaussian}
\overline R_T
\ \dot\sim\
N\left(
m_{\mathrm{vM}},
V_{T,(2)}^{\mathrm{vM}}
\right).
\end{equation}

\section{Monte Carlo evaluation of survival and first-passage probabilities}
\label{sec:barrier-survival}

The terminal results considered earlier depend on the correlation path through the
time-averaged correlation.  Barrier events are different: whether an asset crosses
its default boundary before the horizon depends on the full ordered path of
\((X_{1,t},X_{2,t},\varphi_t)\).  We therefore evaluate survival and
first-passage probabilities by simulating the original stochastic-correlation
model \eqref{eq:corrected-log-asset-dynamics}, without replacing \(R_t\) by a
constant or by \(\overline R_T\).

Let
\[
b_i=\log B_i,
\qquad
\tau_i=\inf\{t\geq0:X_{i,t}\leq b_i\},
\qquad i=1,2,
\]
and let \(j\in\{\mathrm{CBM},\mathrm{vM}\}\) index the Circular Brownian
motion and von Mises process specifications for \(\varphi\), respectively.  We write
\(\mathbb P_j\) for probabilities under specification \(j\).  The joint
survival probability, the first-to-default probability, and the probability that
both names have defaulted by time \(T\) are
\begin{align}
 p_{\mathrm{Surv}}^{(j)}(T)
 &:=
 \mathbb P_j(\tau_1>T,\tau_2>T),
 \label{eq:mc-target-survival}\\
 p_{\mathrm{FtD}}^{(j)}(T)
 &:=
 \mathbb P_j(\min\{\tau_1,\tau_2\}\leq T)
 =1-p_{\mathrm{Surv}}^{(j)}(T),
 \label{eq:mc-target-first-default}\\
 p_{\mathrm{JFPT}}^{(j)}(T)
 &:=
 \mathbb P_j(\tau_1\leq T,\tau_2\leq T)
 =\mathbb P_j(\max\{\tau_1,\tau_2\}\leq T).
 \label{eq:mc-target-joint-fpt}
\end{align}
Thus \(p_{\mathrm{FtD}}^{(j)}\) concerns the first default, whereas
\(p_{\mathrm{JFPT}}^{(j)}\) is the joint first-passage probability used in this
paper.

For each model \(j\), simulate \(M\) independent numerical paths of the state
process
\[
\left\{
(X_{1,t}^{(j,m)},X_{2,t}^{(j,m)},\varphi_t^{(j,m)}):0\leq t\leq T
\right\},
\qquad m=1,\ldots,M,
\]
using the corresponding angular dynamics, \(R_t^{(j,m)}=
\cos^2(\varphi_t^{(j,m)})\), a common systematic Brownian motion, and
independent idiosyncratic Brownian motions for the two assets on each path.
Let
\[
\tau_i^{(j,m)}
=
\inf\{t\geq0:X_{i,t}^{(j,m)}\leq b_i\}
\]
be the resulting first-passage times.  The Monte Carlo estimators are
\begin{align}
\widehat p_{\mathrm{Surv},M}^{(j)}(T)
&=
\frac1M\sum_{m=1}^M
\mathbf 1\{\tau_1^{(j,m)}>T,\tau_2^{(j,m)}>T\},
\label{eq:mc-estimator-survival}\\
\widehat p_{\mathrm{FtD},M}^{(j)}(T)
&=
\frac1M\sum_{m=1}^M
\mathbf 1\{\min(\tau_1^{(j,m)},\tau_2^{(j,m)})\leq T\}
=1-\widehat p_{\mathrm{Surv},M}^{(j)}(T),
\label{eq:mc-estimator-first-default}\\
\widehat p_{\mathrm{JFPT},M}^{(j)}(T)
&=
\frac1M\sum_{m=1}^M
\mathbf 1\{\tau_1^{(j,m)}\leq T,\tau_2^{(j,m)}\leq T\}.
\label{eq:mc-estimator-joint-fpt}
\end{align}
More generally, the empirical joint distribution of the two passage times is
\begin{equation}
\widehat F_{\tau_1,\tau_2;M}^{(j)}(s,t)
=
\frac1M\sum_{m=1}^M
\mathbf 1\{\tau_1^{(j,m)}\leq s,\tau_2^{(j,m)}\leq t\},
\qquad s,t\geq0.
\label{eq:mc-estimator-joint-fpt-cdf}
\end{equation}

\begin{prop}\label{prop:mc-survival-fpt}
For fixed \(T\) and \(j\in\{\mathrm{CBM},\mathrm{vM}\}\), the estimators
\eqref{eq:mc-estimator-survival} and
\eqref{eq:mc-estimator-joint-fpt} are unbiased and strongly consistent.  If
\(p\) denotes the corresponding probability, then
\begin{equation}
\operatorname{Var}(\widehat p_M)
=
\frac{p(1-p)}{M},
\qquad
\sqrt M(\widehat p_M-p)
\xrightarrow{D}
N\bigl(0,p(1-p)\bigr).
\label{eq:mc-estimator-clt}
\end{equation}
A consistent estimator of the Monte Carlo standard error is
\begin{equation}
\widehat{\operatorname{se}}(\widehat p_M)
=
\sqrt{\frac{\widehat p_M(1-\widehat p_M)}{M}}.
\label{eq:mc-standard-error}
\end{equation}
The same statements hold pointwise for
\(\widehat F_{\tau_1,\tau_2;M}^{(j)}(s,t)\).
\end{prop}

\begin{proof}
Each estimator is the sample mean of independent Bernoulli random variables
whose success probability is the target probability.  Unbiasedness and the
variance formula follow directly.  Strong consistency follows from the strong
law of large numbers, and \eqref{eq:mc-estimator-clt} follows from the classical
central limit theorem.  Replacing \(p\) by its consistent sample analogue gives
\eqref{eq:mc-standard-error}.
\end{proof}

\section{Simulation study}\label{sec:simulation}

The numerical study reports four distinct probability measures. At horizon \(T\), terminal joint default is
\begin{equation}
p_{\mathrm{JD}}(T)
=
\mathbb P\{S_{1,T}\leq B_1,S_{2,T}\leq B_2\}.
\label{eq:simulation-jd}
\end{equation}
By Corollary \ref{cor:joint-default-horizon-corrected}, its conditional value
is \(\Phi_2(d_1,d_2;\overline R_T)\), so we estimate it by the
average across simulated correlation paths
\[
\widehat p_{\mathrm{JD}}(T)
=
\frac{1}{M}\sum_{m=1}^{M}
\Phi_2\!\left(d_1,d_2;\overline R_T^{(m)}\right).
\]
The remaining quantities are path events: joint survival
\(p_{\mathrm{Surv}}(T)\), first to default
\(p_{\mathrm{FtD}}(T)=1-p_{\mathrm{Surv}}(T)\), and joint first passage
\(p_{\mathrm{JFPT}}(T)\), as defined in
\eqref{eq:mc-target-survival}--\eqref{eq:mc-target-joint-fpt}.  In
particular, \(p_{\mathrm{JFPT}}\) is the probability that both barriers
have been crossed by \(T\).  An
asset may cross its barrier and subsequently recover, so
\(p_{\mathrm{JD}}\) and \(p_{\mathrm{JFPT}}\) need not coincide.

The reference parameter set is $S_{1,0}=S_{2,0}=100, B_1=B_2=60,\mu_1=\mu_2=0.03, \sigma_1=\sigma_2=0.25.$
For the von Mises process, $\varphi_0=\mu_\varphi=\arccos\sqrt{0.20}=1.1071,\lambda=1.96,\sigma_\varphi=0.70$,
and hence \(\kappa=2\lambda/\sigma_\varphi^2=8\).  We use
\(T\in\{0.25,0.5,1,2\}\) years, \(M=30{,}000\) paths, and increments
\(\Delta t=1/504\). The correlation paths are generated using Euler-Maruyama discretization. 

\subsection{Dependence benchmarks across horizons}
\label{subsec:simulation-benchmarks}

Table~\ref{tab:simulation-benchmark} reports the two year results.  Under the
von Mises process at the reference parameter set, the joint default probability is
$0.013$, survival probability is $0.745$, first time to default probability is $0.255$, and joint first passage time probability is
$0.036$.  The mean time averaged correlation is $0.261$. Circular Brownian motion produces the larger terminal joint default
estimate, accompanied by a higher mean time averaged correlation $0.366$, so the
difference is not a comparison at a common dependence level. Figure~\ref{fig:simulation-horizon} shows the four measures across time.

\begin{table}[H]
\centering
\caption{Two year simulation benchmarks}
\label{tab:simulation-benchmark}
\begin{adjustbox}{max width=\textwidth}
\begin{tabular}{lccccc}
\toprule
Dependence specification &
\(\overline R_T\), mean (SD) &
\(p_{\mathrm{JD}}\) &
\(p_{\mathrm{Surv}}\) &
\(p_{\mathrm{FtD}}\) &
\(p_{\mathrm{JFPT}}\)\\
\midrule
Constant correlation
& 0.270 (0.000) & 0.01271 (exact) & 0.74420 [0.73923, 0.74911]
& 0.25580 [0.25089, 0.26077] & 0.03473 [0.03272, 0.03687]\\
Circular Brownian motion
& 0.366 (0.225) & 0.01776 [0.01764, 0.01787] & 0.75333 [0.74842, 0.75818]
& 0.24667 [0.24182, 0.25158] & 0.04263 [0.04041, 0.04498]\\
von Mises process
& 0.261 (0.136) & 0.01292 [0.01287, 0.01298] & 0.74490 [0.73994, 0.74980]
& 0.25510 [0.25020, 0.26006] & 0.03643 [0.03437, 0.03861]\\
\bottomrule
\end{tabular}
\end{adjustbox}
\begin{minipage}{0.98\textwidth}
\footnotesize
\emph{Notes:} \(M=30{,}000\), \(\Delta t=1/504\), and \(T=2\).
Square brackets show 95\% intervals.  The constant model terminal
probability is analytical at its fixed dependence.
\end{minipage}
\end{table}

\begin{figure}[H]
\centering
\includegraphics[width=0.96\textwidth]{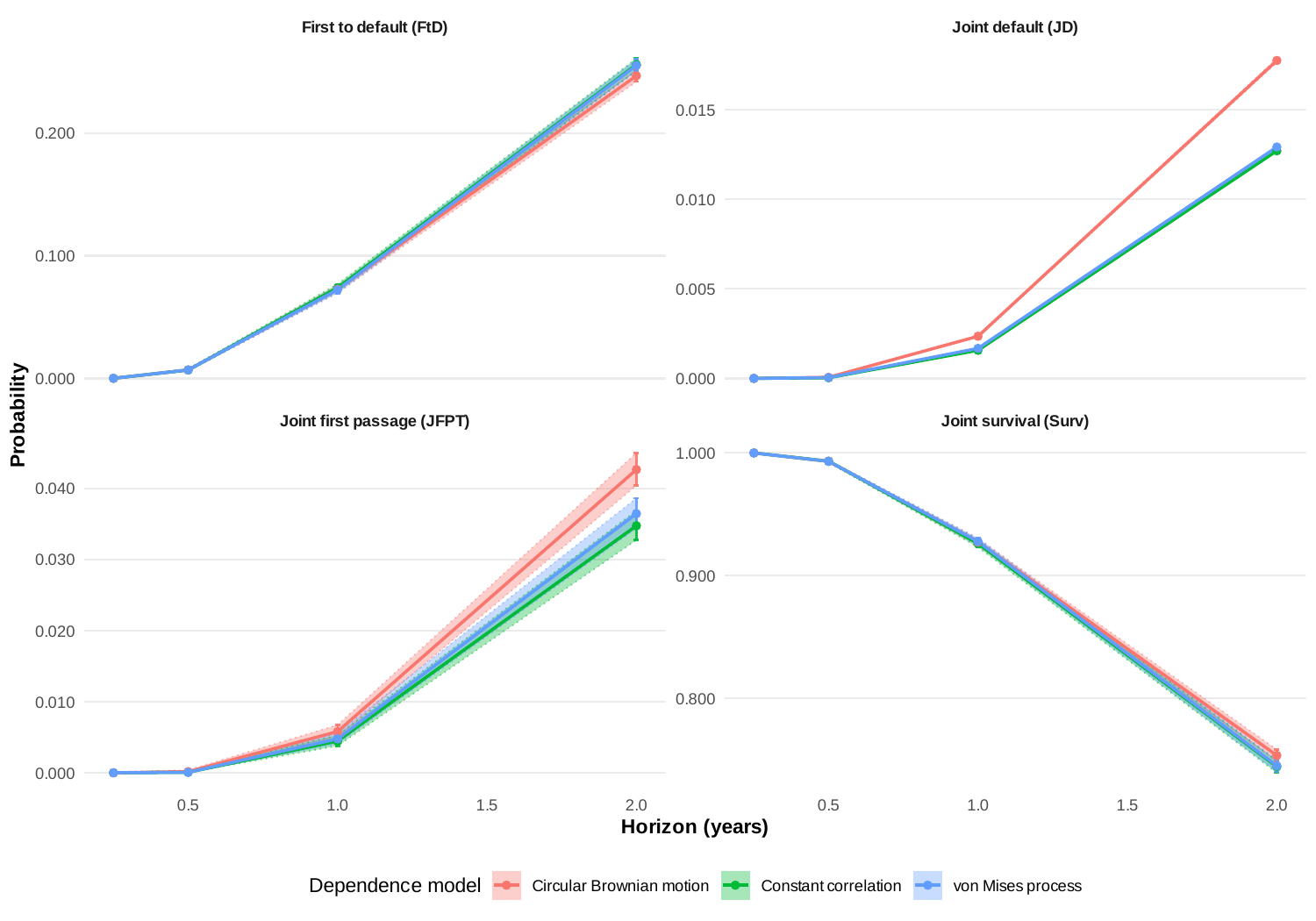}
\caption{Credit event probabilities by dependence process: terminal joint
default, joint survival, first to default, and joint first passage under
constant dependence, Circular Brownian motion, and the von Mises process.
The figure shows 95\% intervals.}
\label{fig:simulation-horizon}
\end{figure}

\subsection{Effects of model parameters}
\label{subsec:simulation-sensitivity}

Marginal asset volatility generates the largest changes, see Table \ref{tab:simulation-sensitivity}.  We see that raising
$\sigma_1,\sigma_2$ expectedly increases terminal joint default probability, reduces survival probability, and increases joint first passage time probability.  

\begin{table}[H]
\centering
\caption{Two year effects of asset volatility and initial dependence}
\label{tab:simulation-sensitivity}
\begin{adjustbox}{max width=\textwidth}
\begin{tabular}{llccccc}
\toprule
Experiment & Level & Mean \(\overline R_T\) &
\(p_{\mathrm{JD}}\) & \(p_{\mathrm{Surv}}\) &
\(p_{\mathrm{FtD}}\) & \(p_{\mathrm{JFPT}}\)\\
\midrule
Asset volatility \(\sigma\) & 0.15 & 0.261 & 0.000184 & 0.97957 & 0.02043 & 0.000467\\
& 0.25 & 0.261 & 0.01292 & 0.74490 & 0.25510 & 0.03643\\
& 0.35 & 0.261 & 0.05433 & 0.47547 & 0.52453 & 0.14350\\
\addlinespace
Initial dependence \(R_0\) & 0.10 & 0.238 & 0.01215 & 0.74333 & 0.25667 & 0.03383\\
& 0.50 & 0.323 & 0.01514 & 0.74963 & 0.25037 & 0.04213\\
& 0.90 & 0.423 & 0.01915 & 0.76190 & 0.23810 & 0.05400\\
\bottomrule
\end{tabular}
\end{adjustbox}
\end{table}

The initial correlation $(R_0)$ experiment illustrates a shift toward concordant
outcomes. Higher positive initial correlation makes the two names more likely
to experience the same outcome, i.e., they either both cross their barriers or
both survive. Since the dependence specification changes only the coupling between
the two asset processes, it leaves each asset's marginal dynamics
unchanged, the individual first-passage probabilities
\(P(\tau_i\leq T)\) remain fixed. Therefore, an increase in the
joint first-passage probability reduces the first-to-default probability
and raises joint survival. Thus probability is reallocated from the discordant
events in which only one name defaults toward the concordant events
in which both default or both survive.

Increasing \(\sigma_\varphi\), the volatility of the correlation
process raises the mean time-averaged correlation.  Terminal joint
default, joint survival, and joint first passage probabilities increase,
whereas the first-to-default probability decreases.  Increasing
\(\kappa\) while holding \(\sigma_\varphi\) fixed also increases the
mean-reversion parameter \(\lambda\).  The resulting increase in
stationary concentration reduces the mean time-averaged correlation,
terminal joint default, joint survival, and joint first passage
probabilities, while increasing the first-to-default probability.
Finally, increasing \(\lambda\) while adjusting \(\sigma_\varphi\) to
hold \(\kappa\) fixed raises the finite-horizon mean time-averaged
correlation because the process mixes more rapidly away from its common
initial state.  Terminal joint default and joint first passage
probabilities generally increase, joint survival increases, and the
first-to-default probability decreases.  The simulations show that \(\lambda\) controls the speed of mean
reversion, whereas \(\kappa\) controls the stationary concentration of
the angular process.
\begin{figure}[H]
\centering
\includegraphics[width=0.94\textwidth]{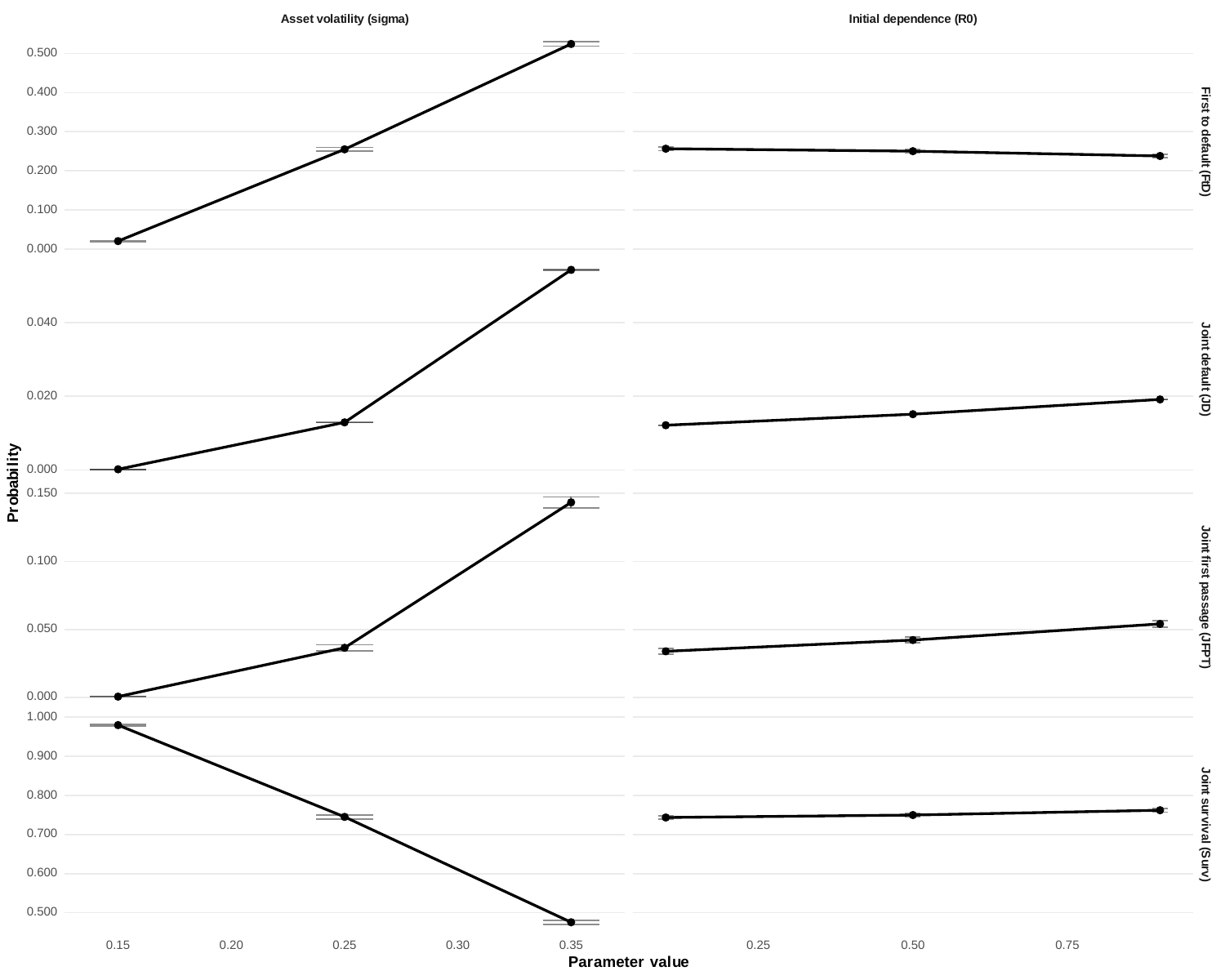}
\caption{Effects of asset volatility and initial dependence under the von Mises process.
The asset-volatility and initial-dependence experiments report all four
terminal and barrier event probabilities; 95\% intervals are shown.}
\label{fig:simulation-sensitivity}
\end{figure}

\begin{figure}[H]
\centering
\includegraphics[width=0.98\textwidth]{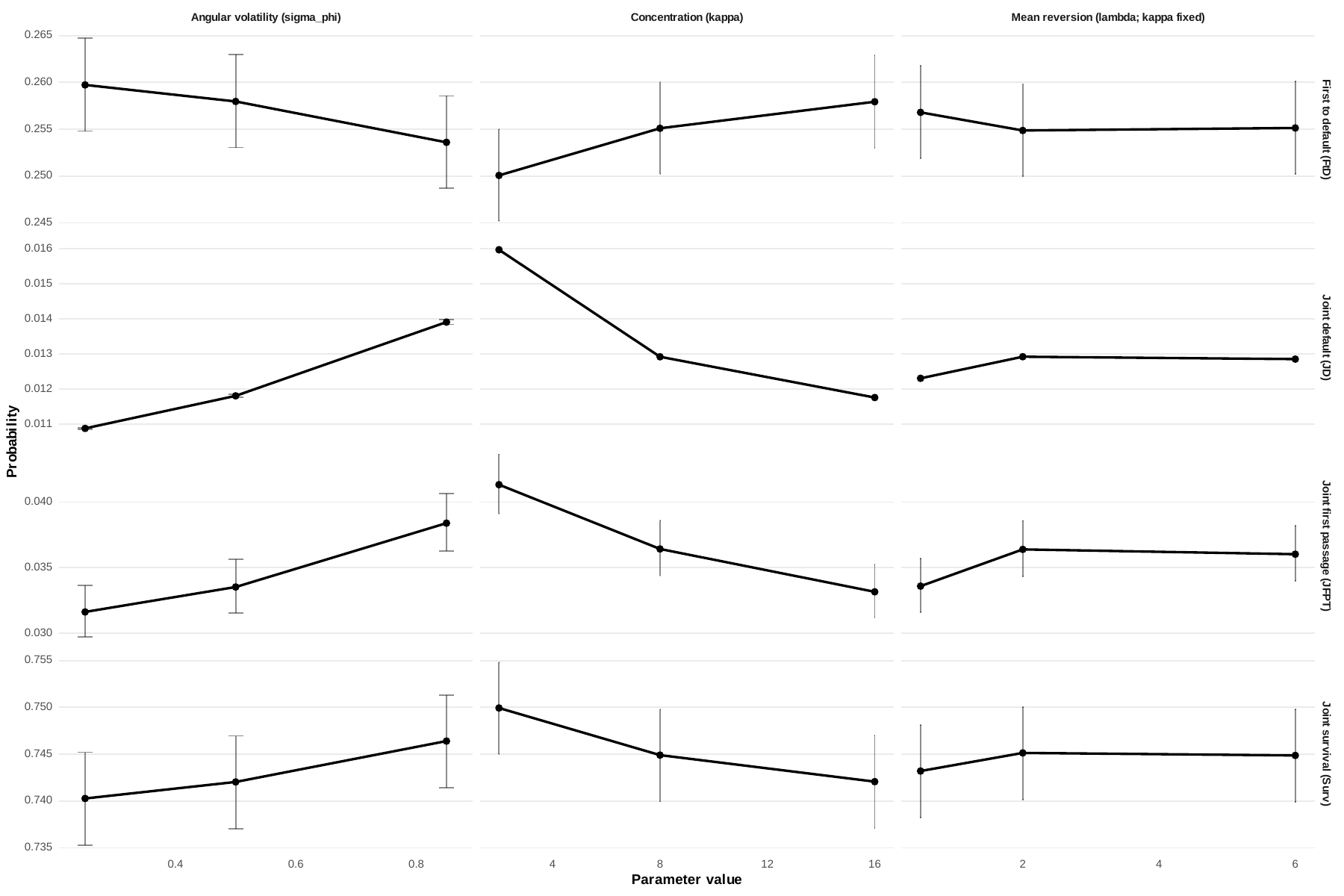}
\caption{Effects of angular volatility, concentration, and mean-reversion
speed under the von Mises process.  The concentration experiment holds
\(\sigma_\varphi\) fixed, while the mean-reversion experiment holds
\(\kappa\) fixed.  The figure shows 95\% intervals.}
\label{fig:simulation-sensitivity-dynamics}
\end{figure}

Figure~\ref{fig:simulation-sensitivity-mean-correlation}
demonstrates
the effect of correlation parameter on its effect on dependence.

\begin{figure}[H]
\centering
\includegraphics[width=0.98\textwidth]{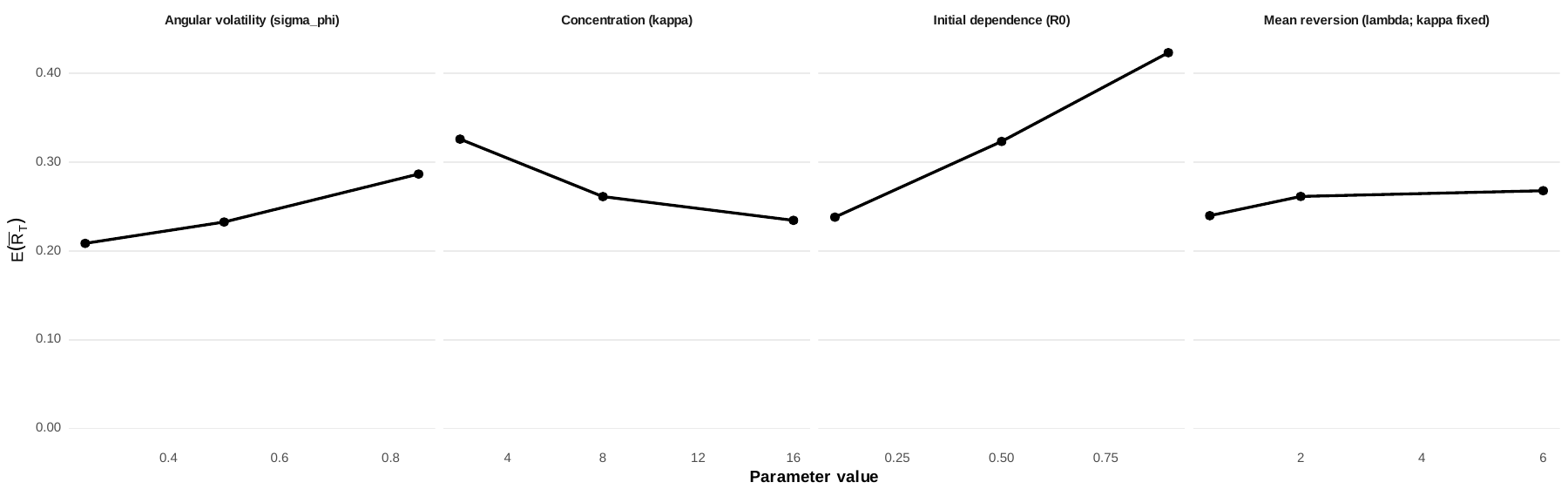}
\caption{Effects of correlation process parameters on mean two year
time averaged correlation under the von Mises process.  Each point is the sample
mean of \(\overline R_T\) across \(30{,}000\) paths.}
\label{fig:simulation-sensitivity-mean-correlation}
\end{figure}

\subsection{Portfolio VaR and expected shortfall}

\label{subsec:simulation-var-es}

For an infinitely granular portfolio with unit loss given default, let
\(L_T\) denote the terminal portfolio-loss fraction. We summarize its upper
tail by
\(\operatorname{VaR}_{\alpha}(L_T)\), the \(\alpha\)-quantile of \(L_T\), and
\[
\operatorname{ES}_{\alpha}(L_T)
=
\mathbb E\!\left[
L_T\mid L_T\geq \operatorname{VaR}_{\alpha}(L_T)
\right].
\]
To obtain these quantities, we combine the finite-horizon approximation to the
law of \(\overline R_T\) developed in
Section~\ref{sec:time-averaged-correlation} with the conditional Vasicek loss
distribution. Since a Gaussian approximation can take values outside the correlation interval, it is truncated and renormalized on \((0,1)\). Denoting
the resulting density by \(\widetilde h_T\), the approximate terminal-loss CDF
is
\[
\widetilde F_{L_T}(x)
=
\int_0^1
F_{\mathrm V}(x;p,\rho)\,
\widetilde h_T(\rho)\,d\rho,
\qquad 0<x<1.
\]
The corresponding VaR is obtained by numerical inversion of
\(\widetilde F_{L_T}\), and expected shortfall is evaluated by numerical
integration over the loss tail.

Figure~\ref{fig:var-es-validation} compares the approximate VaR with a direct
Monte Carlo benchmark at the 95\%, 99\%, and 99.9\% levels over the four
horizons. Points on the dashed line indicate exact agreement. For the von
Mises process, the approximation is already reasonably accurate at the shorter
horizons and becomes particularly close from one year onward. This improvement
is consistent with the averaging effect in \(\overline R_T\): as the horizon
increases, the distribution of the additive functional becomes smoother and
the Gaussian approximation becomes more informative.

The Circular Brownian motion approximation displays larger short-horizon
discrepancies, most visibly at the 99\% level. Under its stationary law, the
instantaneous correlation has substantial mass near the boundaries, while a
short-horizon average remains bounded and can be markedly non-Gaussian. A
truncated normal distribution captures the support restriction but not all of
this boundary behaviour. At the 99.9\% level, several simulated and approximate
VaR values are close to the maximal loss of one. 

\begin{figure}[H]
\centering
\includegraphics[width=0.92\textwidth]{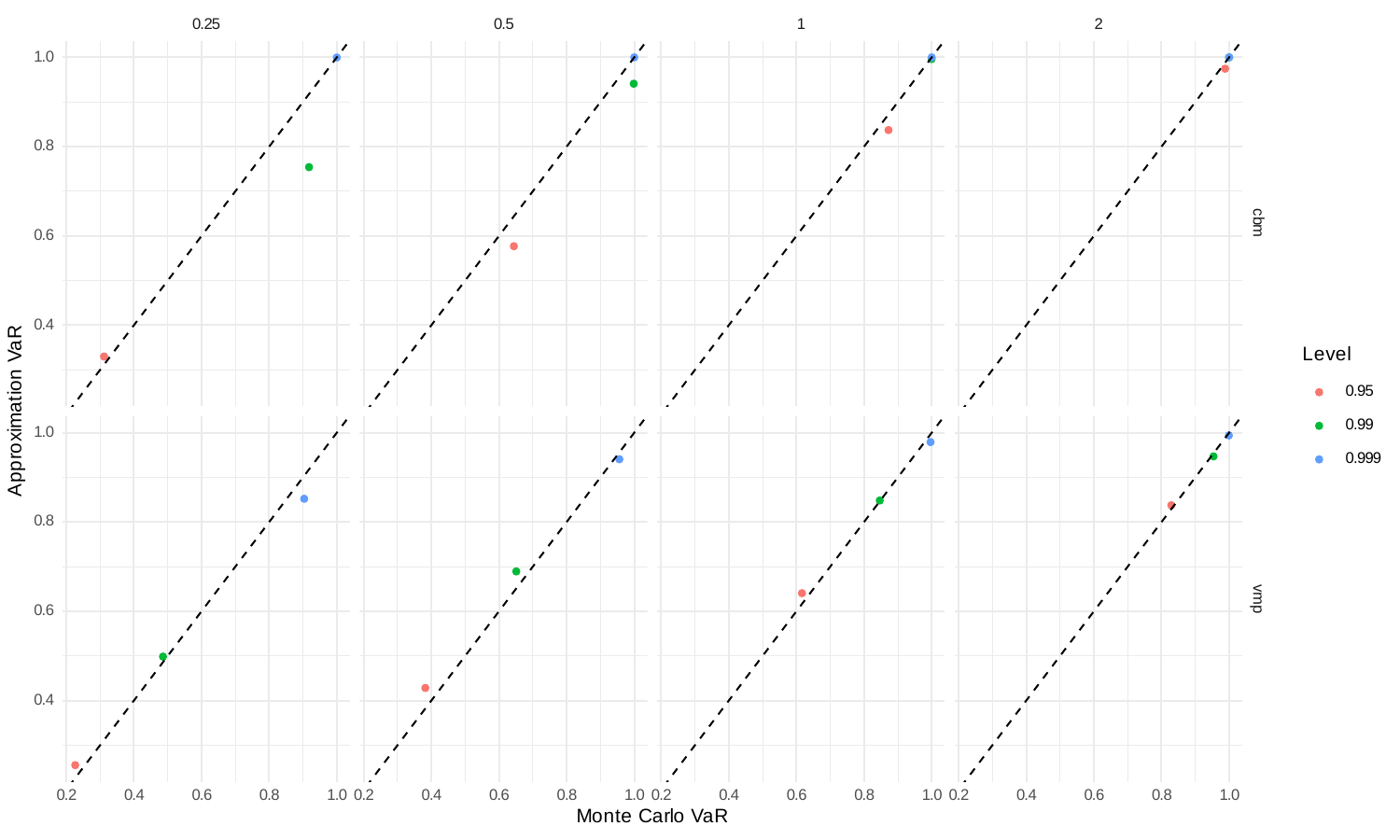}
\caption{Approximate and Monte Carlo terminal portfolio loss VaR under
Circular Brownian motion and the von Mises process. Points report the 95\%,
99\%, and 99.9\% levels for the simulation baseline; the dashed line denotes
equality.Points correspond to the 95\%, 99\%, and 99.9\% VaR levels under the baseline simulation, while the dashed 45-degree line indicates exact agreement between the approximation and Monte Carlo estimates.}
\label{fig:var-es-validation}
\end{figure}

\section{Empirical illustration: U.S. bank charge-off rates}
\label{sec:empirical}

This section applies the stochastic-correlation framework to quarterly U.S.
bank charge-off rates. For each loan category, we estimate a static Vasicek
model, Circular Brownian motion, and the von Mises process using the Vasicek
loss density as a quasi-likelihood. We compare their in-sample fit, examine the
smoothed correlation paths implied by the dynamic models, and use the fitted
von Mises specification to calculate illustrative two-year credit-event
probabilities. Since net charge-offs are accounting loss flows rather than
obligor-level default frequencies, the fitted states are interpreted as
effective category-level correlations within the observation model.

We use the Federal Reserve Board release \emph{Charge-Off and Delinquency
Rates on Loans and Leases at Commercial Banks}.\footnote{Data source:
\url{https://www.federalreserve.gov/releases/chargeoff/}.} The same data source
has previously been studied in a credit-risk setting by
\cite{metzler2020state}. The sample contains quarterly net charge-off rates for
eleven loan categories through the third quarter of 2025. Published rates are
annualized percentages, which we convert to quarterly fractions by dividing by
\(400\).

Net charge-offs combine newly recognized charge-offs with recoveries on loans
charged off in earlier periods. They are therefore accounting loss flows, not
direct observations of obligor default frequencies. Some reports are
nonpositive because recoveries exceed contemporaneous charge-offs. We retain
the economic information in these observations rather than recoding them as
zero. Since the Vasicek loss density has support \((0,1)\), the baseline
quasi-likelihood treats a nonpositive report as indicating a model-implied loss
fraction no greater than
\[
\delta_c
=
\frac{0.01}{100\times4\times2}
=
1.25\times10^{-5},
\]
which is half the smallest quarterly increment implied by the reporting
precision. As robustness checks, we alternatively omit these observations or
replace them by \(\delta_c\). Accordingly, the fitted state should be viewed as
an effective correlation measure for category-level loss outcomes, rather than
as a direct estimate of obligor-level asset correlation.

\subsection{Vasicek quasi-likelihood and estimation}
\label{subsec:empirical-likelihood}

\paragraph{Model specifications.}

For each loan category \(j\), we estimate a static Vasicek model and two
stochastic-correlation models. In the static specification,
\[
\vartheta_{j,\mathrm{static}}=(p_j,\rho_j),
\]
where \(p_j\) is the mean quarterly loss rate and \(\rho_j\) is a constant
correlation parameter. Under Circular Brownian motion,
\[
d\varphi_{j,t}
=
\sigma_{\varphi,j}\,dB_{j,t},
\qquad
\vartheta_{j,\mathrm{CBM}}
=
(p_j,\sigma_{\varphi,j}).
\]
Under the von Mises specification,
\[
d\varphi_{j,t}
=
-\lambda_j\sin(\varphi_{j,t}-\mu_{\varphi,j})\,dt
+
\sigma_{\varphi,j}\,dB_{j,t},
\]
with parameter vector
\[
\vartheta_{j,\mathrm{vM}}
=
(p_j,\lambda_j,\sigma_{\varphi,j},\mu_{\varphi,j}).
\]
Both dynamic models use
\[
R_{j,t}=\cos^2(\varphi_{j,t}).
\]
The Vasicek density is used as a quasi-likelihood observation model; the
specification does not require each loan category to be a literal homogeneous
portfolio of obligor defaults.

\paragraph{Observation contributions.}

Let \(c_{j,k}\) denote the reported quarterly loss fraction for category \(j\)
in quarter \(k\), and let \(\Delta=0.25\) years. For either stochastic-
correlation specification, define the quarterly time-averaged correlation by
\[
\overline R_{j,k}
=
\frac{1}{\Delta}
\int_{(k-1)\Delta}^{k\Delta}
\cos^2(\varphi_{j,s})\,ds.
\]
Applying Proposition~\ref{prop:vasicek-density-stochastic-correlation} over the
\(k\)th quarter, the model-implied loss fraction, conditional on the
correlation path, has the Vasicek distribution with parameters \(p_j\) and
\(\overline R_{j,k}\). Its density is given by
\eqref{eq:vasicek-density-kernel}.

A strictly positive report is treated as an exact realization of this
continuous loss variable and contributes
\[
f(c_{j,k};p_j,\overline R_{j,k})
\]
to the quasi-likelihood. In the baseline specification, a nonpositive report
is represented by the event that the model-implied loss fraction does not
exceed \(\delta_c\). Its contribution is therefore
\[
\begin{aligned}
F_{\mathrm V}(\delta_c;p_j,\overline R_{j,k})
&=
\int_0^{\delta_c}
f(x;p_j,\overline R_{j,k})\,dx \\
&=
\Phi\left(
\frac{
\sqrt{1-\overline R_{j,k}}\,\Phi^{-1}(\delta_c)
-
\Phi^{-1}(p_j)
}{
\sqrt{\overline R_{j,k}}
}
\right).
\end{aligned}
\]
Thus the observation contribution at a generic correlation level \(\rho\) is
\[
g_{j,k}(\rho;p_j)
=
\begin{cases}
f(c_{j,k};p_j,\rho), & c_{j,k}>0,\\[4pt]
F_{\mathrm V}(\delta_c;p_j,\rho), & c_{j,k}\leq0.
\end{cases}
\]
For a dynamic model this expression is evaluated at
\(\rho=\overline R_{j,k}\); for the static Vasicek model it is evaluated at
\(\rho=\rho_j\).

\paragraph{Particle-filter estimation.}

For the stochastic-correlation specifications, each quarterly observation
contribution must be averaged over the conditional distribution of the
unobserved angular path. This integral is unavailable in closed form, so we
approximate it using a bootstrap particle filter
\parencite{kitagawa1996monte,doucet2011tutorial}.

Let \(\mathcal F_{j,k-1}\) denote the information in the observations for
category \(j\) through quarter \(k-1\). Under model \(\mathcal M\), the
predictive contribution of observation \(k\) is
\[
L_{j,k,\mathcal M}(\vartheta_{j,\mathcal M})
=
\mathbb E_{\vartheta_{j,\mathcal M}}
\left[
g_{j,k}\!\left(\overline R_{j,k};\vartheta_{j,\mathcal M}\right)
\middle|
\mathcal F_{j,k-1}
\right].
\]
The particle filter approximates the predictive distribution of
\(\overline R_{j,k}\) by
\[
\left\{
\overline R_{j,k}^{(m)},
w_{j,k|k-1}^{(m)}
\right\}_{m=1}^{N},
\qquad
\sum_{m=1}^{N}w_{j,k|k-1}^{(m)}=1.
\]
Consequently,
\[
\widehat L_{j,k,\mathcal M}(\vartheta_{j,\mathcal M})
=
\sum_{m=1}^{N}
w_{j,k|k-1}^{(m)}
g_{j,k}\!\left(
\overline R_{j,k}^{(m)};
\vartheta_{j,\mathcal M}
\right),
\]
and the particle approximation to the marginal quasi-log-likelihood is
\begin{equation}
\widehat l_{j,\mathcal M}(\vartheta_{j,\mathcal M})
=
\sum_{k=1}^{K_j}
\log\widehat L_{j,k,\mathcal M}(\vartheta_{j,\mathcal M}).
\label{eq:particle-likelihood}
\end{equation}

Within each quarter, particles evolve according to the relevant circular
diffusion, and \(\overline R_{j,k}^{(m)}\) is computed from the simulated
within-quarter path. The contribution \(g_{j,k}\) is then used to update and
normalize the particle weights. Systematic resampling is performed whenever
\[
\operatorname{ESS}_{j,k}
=
\left\{
\sum_{m=1}^{N}
\bigl(w_{j,k}^{(m)}\bigr)^2
\right\}^{-1}
<
\frac{N}{2}.
\]

For Circular Brownian motion, initial particles are drawn from the uniform
invariant distribution on the circle. For the von Mises process, they are
drawn from its invariant distribution with concentration parameter $\kappa_j
=
\frac{2\lambda_j}{\sigma_{\varphi,j}^{2}}$.
Particles in both models evolve on the full circle. Since
\(R=\cos^2\varphi\) is invariant under several transformations of
\(\varphi\), distinct values of \(\mu_{\varphi,j}\) can generate
observationally equivalent specifications. We therefore impose
\(\mu_{\varphi,j}\in[0,\pi/2]\) to select a single representative from these
equivalent parameterizations.

The numerical search is conducted over $p_j\in[10^{-6},0.25],
\sigma_{\varphi,j}\in[0.05,5]$
and, for the von Mises process, $
\lambda_j\in[0.05,10],
\mu_{\varphi,j}\in[0,\pi/2]$.
Each quarterly transition is approximated using 20 Euler--Maruyama substeps,
and each likelihood evaluation during optimization uses 1,500 particles. We
employ a bounded Nelder--Mead search with six multistarts and fixed common
random numbers, followed by a continuation run from the best solution. The
fitted parameters are reevaluated using 6,000 particles. State summaries use
2,500 trajectories obtained by sampling terminal particles and tracing their
recorded ancestors backward through the particle genealogy
\parencite{kitagawa1996monte,doucet2011tutorial}.

\paragraph{Model comparison.}

Let \(\widetilde l_{j,\mathcal M}\) denote the 6,000-particle reevaluation
of the quasi-log-likelihood at the estimated parameter vector, and let
\(q_{\mathcal M}\) denote the number of estimated parameters. We report
\[
\mathrm{AIC}_{j,\mathcal M}
=
2q_{\mathcal M}-2\widetilde l_{j,\mathcal M},
\qquad
\mathrm{BIC}_{j,\mathcal M}
=
q_{\mathcal M}\log K_j-2\widetilde l_{j,\mathcal M},
\]
where  $q_{\mathrm{static}}=q_{\mathrm{CBM}}=2,
q_{\mathrm{vM}}=4$.
All three models use the same observations and the same treatment of
nonpositive reports. The static criterion is evaluated deterministically,
whereas the dynamic criteria depend on particle approximations. We therefore
repeat the dynamic quasi-likelihood evaluation at each fitted parameter vector
to assess Monte Carlo variation in the reported criteria.

\subsection{Estimated correlation dynamics}
\label{subsec:empirical-correlation}

Table~\ref{tab:empirical-model-comparison} compares the static and dynamic
models using quasi-information criteria. Both stochastic-correlation
specifications produce lower AIC and BIC values than the static Vasicek model
in every category. Among the two dynamic models, BIC favors Circular Brownian
motion in all eleven categories, although the farmland values are nearly tied
at the displayed precision. AIC favors the von Mises process in six categories
and Circular Brownian motion in five.

\begin{table}[H]
\centering
\caption{Quasi-likelihood comparison of static and stochastic correlation models}
\label{tab:empirical-model-comparison}
\begin{adjustbox}{max width=\textwidth}
\begin{tabular}{lrrrrrr}
\toprule
& \multicolumn{2}{c}{Static Vasicek} & \multicolumn{2}{c}{Circular Brownian motion} & \multicolumn{2}{c}{von Mises process}\\
\cmidrule(lr){2-3}\cmidrule(lr){4-5}\cmidrule(lr){6-7}
Category & AIC & BIC & AIC & BIC & AIC & BIC\\
\midrule
Agricultural loans & -1838.3 & -1832.1 & -1935.0 & -1928.8 & -1934.3 & -1921.9\\
All consumer loans & -1552.0 & -1545.8 & -1662.1 & -1655.9 & -1663.1 & -1650.7\\
All real estate loans & -1803.8 & -1797.6 & -1912.3 & -1906.1 & -1911.0 & -1898.6\\
Commercial and industrial & -1754.5 & -1748.3 & -1785.5 & -1779.3 & -1786.1 & -1773.7\\
Commercial real estate & -1419.4 & -1413.6 & -1476.9 & -1471.1 & -1467.6 & -1455.8\\
Credit cards & -1388.8 & -1382.6 & -1451.3 & -1445.1 & -1452.7 & -1440.3\\
Farmland & -1828.2 & -1822.3 & -1842.4 & -1836.6 & -1848.0 & -1836.3\\
Leases & -1928.8 & -1922.6 & -1960.9 & -1954.8 & -1962.9 & -1950.5\\
Other consumer loans & -1780.8 & -1774.6 & -1867.1 & -1860.9 & -1867.1 & -1854.8\\
Residential real estate & -1175.0 & -1169.2 & -1364.0 & -1358.2 & -1359.6 & -1347.9\\
Total loans and leases & -1778.9 & -1772.7 & -1839.3 & -1833.1 & -1817.7 & -1805.4\\
\bottomrule
\end{tabular}
\end{adjustbox}
\begin{minipage}{0.98\textwidth}
\footnotesize
\emph{Notes:} Lower values are preferred.
\end{minipage}
\end{table}

\begin{table}[H]
\centering
\caption{Monte Carlo variability of particle quasi-log-likelihood evaluations}
\label{tab:particle-repeat-variability}
\begin{tabular}{lrr}
\toprule
Category & CBM SD\((\widehat l)\) & von Mises SD\((\widehat l)\)\\
\midrule
Agricultural loans & .15 & .24\\
All consumer loans & .26 & .25\\
All real estate loans & .09 & .07\\
Commercial and industrial & .12 & .02\\
Commercial real estate & .06 & .40\\
Credit cards & .15 & .16\\
Farmland & .13 & .09\\
Leases & .14 & .12\\
Other consumer loans & .12 & .08\\
Residential real estate & .10 & .09\\
Total loans and leases & .09 & .14\\
\bottomrule
\end{tabular}
\begin{minipage}{0.98\textwidth}
\footnotesize
\emph{Notes:} Each value is the standard deviation of four independent
6,000-particle, 20-substep quasi-likelihood evaluations at the reported point
estimate.
\end{minipage}
\end{table}

Table~\ref{tab:particle-repeat-variability} reports the Monte Carlo variation
in repeated particle-likelihood evaluations at the fitted parameters. This
variation is small relative to the differences between the static Vasicek
model and either stochastic-correlation model, so the improvement of the
dynamic models is not attributable to particle-filter noise. For a few
categories, however, the difference between the Circular Brownian motion and
von Mises criteria is of the same order as the repeat variation, and their
relative ranking should therefore be read descriptively. The von Mises fit to
commercial real estate has the largest repeat standard deviation. This reduces
the numerical precision of its reported criteria, but the variation is too
small to change the model ranking for that category.

Table~\ref{tab:static-estimates} reports the static Vasicek quasi-MLEs used in
the information-criterion comparison. The fitted loss and correlation
parameters display different cross-category patterns. Credit cards have the
largest \(\widehat p\), whereas the largest static correlation estimates occur
for residential real estate, commercial real estate, and all real estate
loans. Thus a higher average loss rate need not be accompanied by a higher
fitted correlation.

\begin{table}[H]
\centering
\caption{Static Vasicek quasi-MLEs}
\label{tab:static-estimates}
\begin{tabular}{lrr}
\toprule
Category & \(\widehat p\) & \(\widehat\rho\)\\
\midrule
Agricultural loans & .001045 & .114\\
All consumer loans & .005887 & .017\\
All real estate loans & .001133 & .157\\
Commercial and industrial & .001891 & .054\\
Commercial real estate & .001330 & .227\\
Credit cards & .010629 & .016\\
Farmland & .000262 & .096\\
Leases & .001071 & .057\\
Other consumer loans & .002652 & .018\\
Residential real estate & .001236 & .256\\
Total loans and leases & .002135 & .029\\
\bottomrule
\end{tabular}
\begin{minipage}{0.98\textwidth}
\footnotesize
\emph{Notes:} These estimates are used in the static correlation model
quasi-information criteria reported in
Table~\ref{tab:empirical-model-comparison}.
\end{minipage}
\end{table}

\begin{table}[H]
\centering
\caption{Stochastic correlation quasi-MLEs and smoothed correlation summaries}
\label{tab:empirical-correlation}
\begin{adjustbox}{max width=\textwidth}
\scriptsize
\begin{tabular}{lrrrrrrrrr}
\toprule
Category & CBM \(\widehat p\) & CBM \(\widehat\sigma_\varphi\) & vM \(\widehat p\) &
\(\widehat\lambda\) & \(\widehat\sigma_\varphi\) & \(\widehat\mu_\varphi\) & \(\widehat\kappa\) & vM mean \(\widetilde{\overline R}\) & vM \(q_{0.95}\)\\
\midrule
Commercial real estate & .004281 & .149 & .004407 & .350 & .152 & .963 & 30.36 & .286 & .516\\
Residential real estate & .000475 & .114 & .000485 & .061 & .118 & 1.183 & 8.65 & .184 & .693\\
All real estate loans & .001136 & .082 & .001129 & .108 & .079 & 1.239 & 34.98 & .129 & .330\\
Farmland & .000277 & .121 & .000258 & .924 & .172 & 1.285 & 62.57 & .086 & .160\\
Commercial and industrial & .003175 & .092 & .003162 & .185 & .098 & 1.346 & 38.71 & .078 & .157\\
Agricultural loans & .000761 & .149 & .000739 & .609 & .239 & 1.372 & 21.26 & .074 & .263\\
Leases & .001455 & .062 & .001343 & .322 & .071 & 1.344 & 128.05 & .058 & .116\\
Total loans and leases & .001574 & .064 & .002423 & .311 & .066 & 1.421 & 142.48 & .031 & .066\\
Other consumer loans & .002453 & .061 & .002496 & .512 & .075 & 1.458 & 181.61 & .014 & .048\\
Credit cards & .009335 & .059 & .009359 & .077 & .070 & 1.541 & 31.71 & .012 & .038\\
All consumer loans & .005653 & .056 & .005647 & .070 & .057 & 1.548 & 43.30 & .012 & .042\\
\bottomrule
\end{tabular}
\end{adjustbox}
\end{table}

Table~\ref{tab:empirical-correlation} reports the quasi-MLEs for both correlation
models. For the von Mises process, the last two columns summarize the fitted
sequence of quarterly time-averaged correlations. The reported mean is the
time average of the quarterly smoothed medians, and \(q_{0.95}\) is the
temporal 95th percentile of that sequence. The real estate categories have
the highest average fitted correlations: commercial real estate ranks first,
followed by residential real estate and all real estate loans. Residential
real estate has the largest temporal upper tail, indicating that its fitted
correlation reaches high levels in some quarters even though its average is
below that of commercial real estate. By contrast, the consumer-loan series
have low average fitted correlations. This ordering differs from that of
\(\widehat p_j\), i.e., credit cards have the largest fitted loss parameter but one
of the lowest fitted correlations.

For the von Mises process,
\[
\widehat\kappa_j
=
\frac{2\widehat\lambda_j}{\widehat\sigma_{\varphi,j}^{\,2}}
\]
determines the stationary concentration of the angular process around
\(\widehat\mu_{\varphi,j}\). A larger value implies that
\(\varphi_{j,t}\) is more tightly concentrated around its stationary
location and, through the mapping \(R_{j,t}=\cos^2(\varphi_{j,t})\), produces a
more concentrated fitted correlation distribution around
\(\cos^2(\widehat\mu_{\varphi,j})\). The speed of mean reversion is controlled
separately by \(\widehat\lambda_j\).

Figure~\ref{fig:empirical-correlation} plots the smoothed correlation paths
for the four real estate categories. Circular Brownian motion and the von
Mises process exhibit broadly similar movements over time, although the
estimated levels differ. Commercial real estate shows several periods of
elevated correlation, while residential real estate rises sharply toward the
end of the sample. The increase for all real estate loans is more gradual,
whereas farmland remains at comparatively lower
levels. The late rise
in residential real estate should be interpreted cautiously because its
estimated path is sensitive to the treatment of nonpositive charge-off
observations, as examined next.

\begin{figure}[H]
\centering
\includegraphics[width=0.94\textwidth]{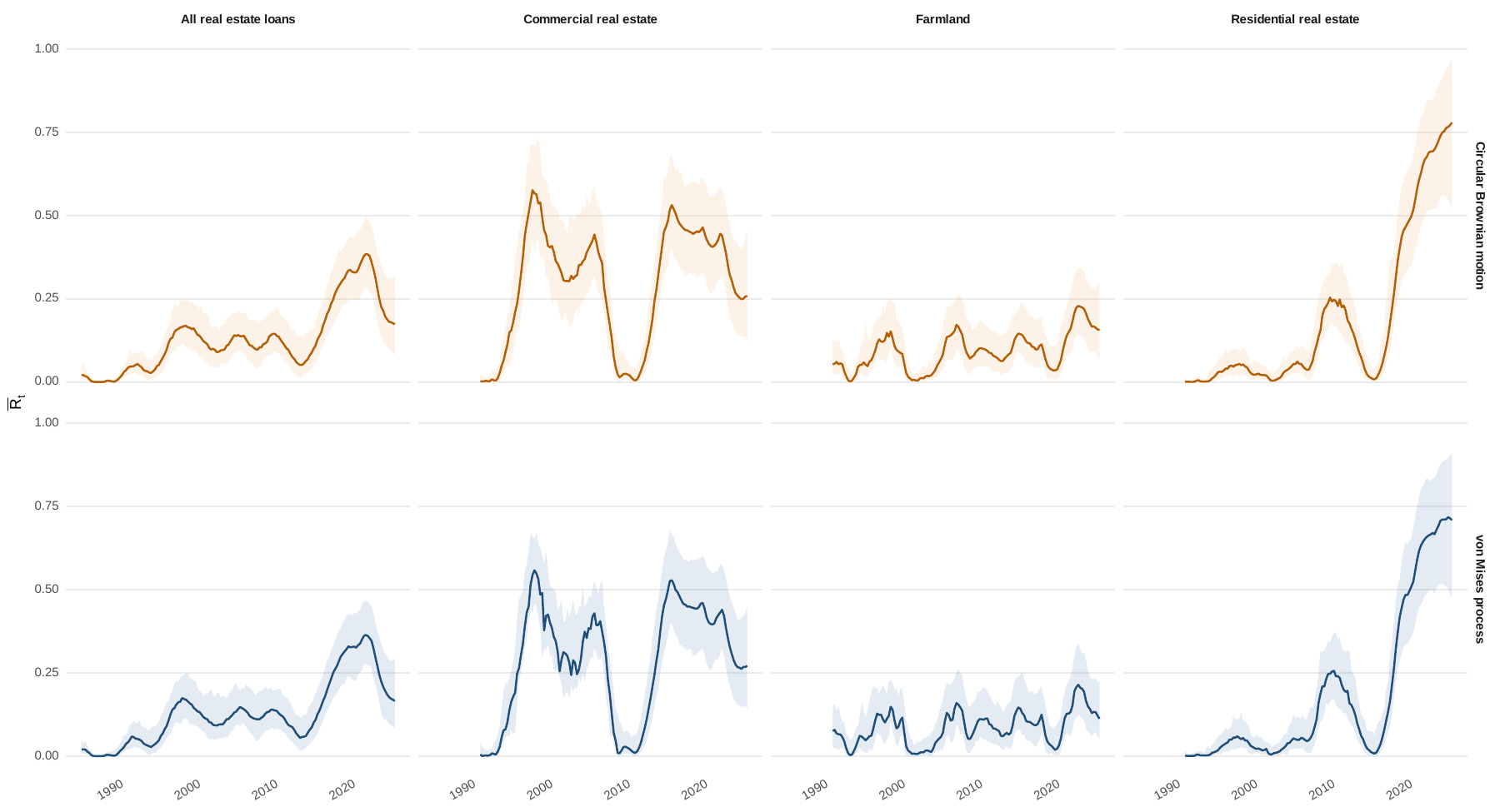}
\caption{Smoothed quarterly time-averaged correlation for real estate
categories under Circular Brownian motion and the von Mises process. The
figure shows pointwise 95\% smoothing intervals conditional on the fitted
parameters.}
\label{fig:empirical-correlation}
\end{figure}

\begin{table}[H]
\centering
\caption{Residential real estate sensitivity to treatment of nonpositive charge-offs}
\label{tab:censoring-robustness}
\begin{adjustbox}{max width=\textwidth}
\begin{tabular}{llrrrrrr}
\toprule
Process & Treatment & \(\widehat p\) & \(\widehat\sigma_\varphi\) &
\(\widehat\lambda\) & \(\widehat\kappa\) & Mean \(\widetilde{\overline R}\) & \(q_{0.95}\)\\
\midrule
Circular Brownian motion & Left censored & .000475 & .114 & --- & --- & .187 & .725\\
Circular Brownian motion & Omit & .000485 & .082 & --- & --- & .098 & .239\\
Circular Brownian motion & Floor & .000481 & .084 & --- & --- & .113 & .295\\
von Mises process & Left censored & .000485 & .118 & .061 & 8.65 & .184 & .693\\
von Mises process & Omit & .000489 & .087 & .085 & 22.60 & .095 & .215\\
von Mises process & Floor & .000489 & .091 & .058 & 13.90 & .113 & .298\\
\bottomrule
\end{tabular}
\end{adjustbox}
\begin{minipage}{0.98\textwidth}
\footnotesize
\emph{Notes:} ``Left censored'' is the baseline specification. ``Floor''
replaces a nonpositive report by \(\delta_c\).
\end{minipage}
\end{table}

Table~\ref{tab:censoring-robustness} compares the baseline left-censoring
treatment with omission and flooring for residential real estate. The
estimated quarterly loss parameter is almost unchanged across the three
treatments, but the fitted correlation dynamics are not. Under both circular
models, omitting or flooring the nonpositive observations produces a lower
average correlation and a substantially lower temporal upper tail than left
censoring. Higher correlation makes the Vasicek loss distribution more
dispersed, placing greater probability on both very low and very high loss
outcomes. Consequently, recovery-dominated reports can support a higher fitted
correlation when they enter the quasi-likelihood as left-censored
observations. The residential correlation path should therefore be
interpreted cautiously because it depends on how nonpositive reports are
treated. The same qualitative pattern appears in the other categories that
contain such observations.

\paragraph{Illustrative two-year probabilities.}

To illustrate how the fitted correlation dynamics translate into credit-event
probabilities, we consider two identical hypothetical exposures for each loan
category. We report the von Mises calculations because this richer
specification allows mean reversion and a nonuniform stationary correlation
distribution; the exercise is illustrative and is not intended as a
model-selection claim. The estimated quarterly loss parameter \(\widehat p_j\)
is treated as a quarterly default-probability proxy and converted into a
two-year marginal probability,
\[
p_{j,H}
=
1-(1-\widehat p_j)^{4H},
\qquad H=2.
\]
We set \(B=100\), \(\mu=0.05\), and \(\sigma=0.15\), and choose the initial
asset-to-barrier ratio as
\[
\log(S_{j,0}/B)
=
-\Phi^{-1}(p_{j,H})\sigma\sqrt H
-
\left(\mu-\frac12\sigma^2\right)H,
\]
so that each exposure has marginal terminal default probability \(p_{j,H}\).
The correlation process is initialized from its fitted 2025Q3 filtering
distribution and evolved under the estimated von Mises dynamics. These
calculations provide a common structural illustration across categories and
should not be interpreted as forecasts for representative firms or loan
portfolios.

\begin{table}[H]
\centering
\caption{Two-year probabilities under the filtered von Mises process}
\label{tab:empirical-scenarios}
\begin{adjustbox}{max width=\textwidth}
\begin{tabular}{lrrrrr}
\toprule
Category & \(p_H\) &
\(p_{\mathrm{JD}}\) & \(p_{\mathrm{Surv}}\) &
\(p_{\mathrm{FtD}}\) & \(p_{\mathrm{JFPT}}\)\\
\midrule
Credit cards                & 0.07247 & 0.005468 & 0.68462 & 0.31538 & 0.02970\\
Commercial real estate      & 0.03472 & 0.004167 & 0.85584 & 0.14416 & 0.01540\\
All consumer loans          & 0.04429 & 0.002104 & 0.80718 & 0.19282 & 0.01138\\
Commercial and industrial   & 0.02501 & 0.000997 & 0.89000 & 0.11000 & 0.004080\\
Residential real estate     & 0.003871 & 0.000884 & 0.98536 & 0.01464 & 0.002420\\
Total loans and leases      & 0.01922 & 0.000443 & 0.91738 & 0.08262 & 0.001920\\
Other consumer loans        & 0.01979 & 0.000425 & 0.91404 & 0.08596 & 0.001840\\
All real estate loans       & 0.008996 & 0.000254 & 0.96284 & 0.03716 & 0.000720\\
Leases                      & 0.01069 & 0.000168 & 0.95284 & 0.04716 & 0.000540\\
Agricultural loans          & 0.005900 & 0.0000666 & 0.97398 & 0.02602 & 0.000220\\
Farmland                    & 0.002064 & 0.0000122 & 0.99214 & 0.00786 & 0.000080\\
\bottomrule
\end{tabular}
\end{adjustbox}
\begin{minipage}{0.98\textwidth}
\footnotesize
\emph{Notes:} The calculations use the fitted von Mises process initialized
from its 2025Q3 filtering distribution. The structural parameters are
\(H=2\), \(B=100\), \(\mu=0.05\), and \(\sigma=0.15\). Probabilities are
estimated using \(M=50{,}000\) paths and monitoring interval
\(\Delta t=1/504\) year. \(p_{\mathrm{JD}}\) is the average conditional
terminal joint-default probability; the remaining columns are path event
frequencies. 
\end{minipage}
\end{table}

Table~\ref{tab:empirical-scenarios}
illustrates
the roles of
marginal risk and correlation. The horizon probability \(p_H\) largely
determines the overall scale of default risk with credit cards having the largest
first-to-default probability, whereas farmland has the smallest. Joint-event
probabilities also reflect the fitted correlation dynamics. Commercial real
estate, for example, has a lower \(p_H\) than all consumer loans but higher
terminal joint-default and joint-first-passage probabilities. Residential
real estate has a relatively high fitted correlation state, but its low
marginal calibration keeps its joint-event probabilities small in absolute
terms.

The quasi-information criteria provide consistent evidence in favor of allowing the correlation state to
vary over time. The fitted paths differ across categories, with the real
estate time series generally displaying higher correlations than the consumer-loan
series. The residential results also show that the inferred path can be
sensitive to the treatment of nonpositive charge-off observations. Finally,
the two-year illustration shows that the
fitted marginal probability determines the overall scale of credit risk,
while the correlation dynamics influence joint default and joint barrier
crossing. 

\subsection{Portfolio VaR and expected shortfall for charge-off data}
\label{subsec:empirical-var-es}

For category \(j\), the fitted quarterly loss parameter is first compounded to
a two-year marginal calibration,
\[
p_{j,2}
=
1-(1-\widehat p_j)^8.
\]
Let \(\widetilde h_{j,2}\) denote the truncated and renormalized approximation
to the stationary two-year average correlation under the fitted process. The
corresponding unconditional terminal-loss distribution is
\[
\widetilde F_{j,2}(x)
=
\int_0^1
F_{\mathrm V}(x;p_{j,2},\rho)\,
\widetilde h_{j,2}(\rho)\,d\rho.
\]
We obtain VaR by numerical inversion and expected shortfall by integrating the
loss distribution beyond the VaR threshold. Figure~\ref{fig:empirical-var-es}
focuses on the 99\% and 99.9\% levels, where the effect of the dependence
specification is most visible.

The two fitted processes imply different long-run tail loss
profiles. For Circular Brownian motion, the invariant angular distribution is
uniform and the stationary mean correlation is therefore one half for every
loan category. The fitted angular volatility changes the persistence and the
finite-horizon dispersion of \(\overline R_T\), but not this stationary mean.
Consequently, the Circular Brownian motion specification imposes a common,
relatively high long-run dependence level, and produces very severe tail losses
for categories with moderate or high compounded loss probabilities. At the
99.9\% level, expected shortfall is close to one for credit cards, all consumer
loans, commercial real estate, and commercial and industrial loans. In the
asymptotic unit-loss-given-default portfolio, this means that conditional on
entering the most adverse 0.1\% of model outcomes, losses are concentrated near
the entire portfolio.

The von Mises process permits category-specific long-run locations,
concentrations, and persistence, and therefore generates more differentiated
and generally more moderate tail-loss estimates. At the 99.9\% level, expected
shortfall is 0.289 for credit cards, 0.219 for commercial and industrial loans,
0.251 for residential real estate, and 0.712 for commercial real estate. The
commercial real estate result is especially informative, i.e., its fitted marginal
loss parameter is below that of credit cards, but its stronger long-run
dependence produces a substantially larger extreme-tail loss. Conversely,
credit cards combine the highest fitted marginal loss rate with comparatively
low fitted correlation, so their tail risk is driven more by marginal credit
quality than by default clustering. Residential real estate illustrates the
opposite phenomena, despite a low marginal calibration, high-correlation
can still generate non-negligible conditional tail losses.
\begin{figure}[H]
\centering
\includegraphics[width=0.94\textwidth]{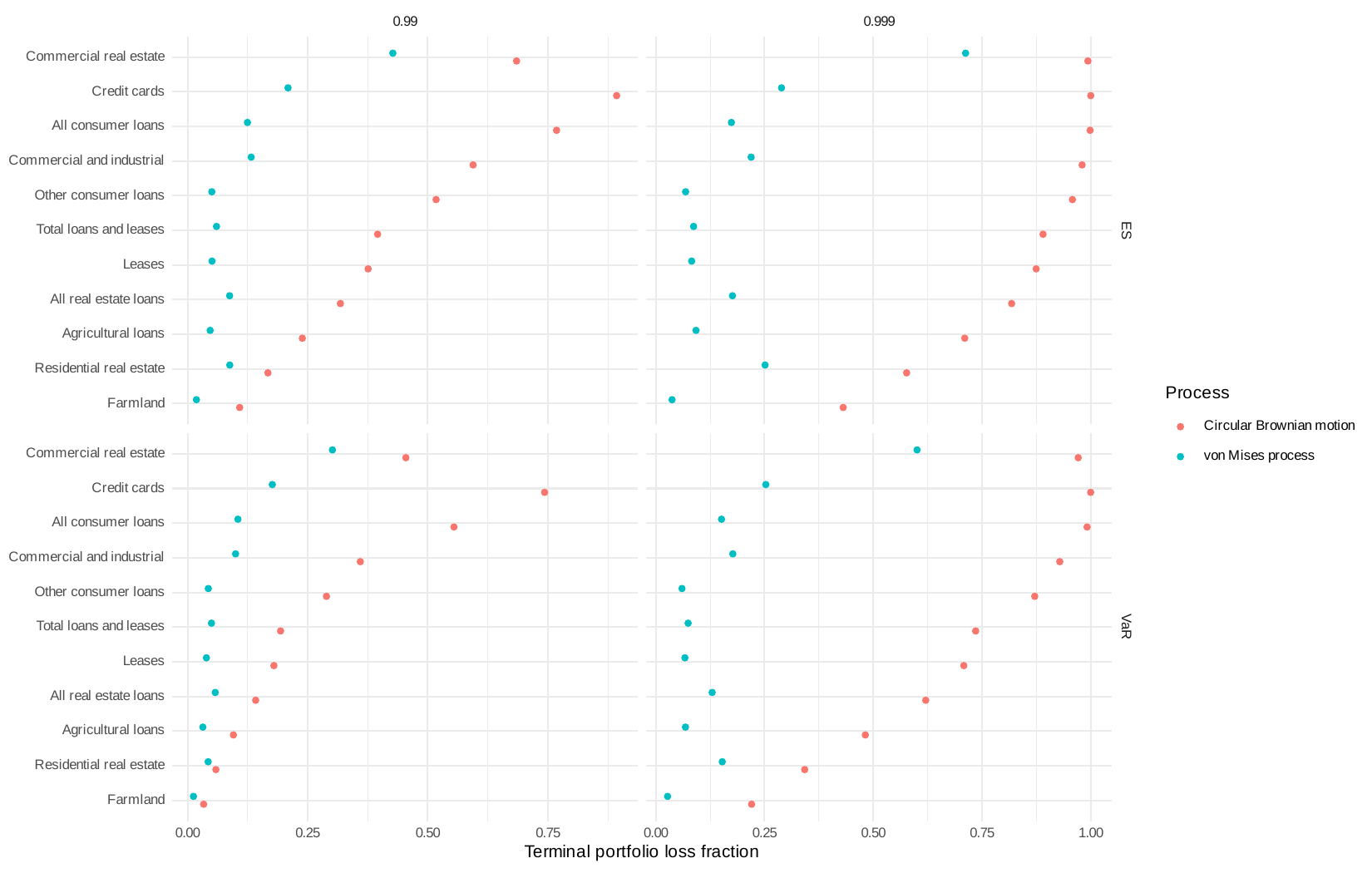}
\caption{Unconditional two-year terminal portfolio-loss VaR and expected
shortfall under the fitted Circular Brownian motion and von Mises processes.
Losses are infinite-granularity fractions with unit loss given default; panels
show the 99\% and 99.9\% levels.}
\label{fig:empirical-var-es}
\end{figure}

\section{Conclusion}\label{sec:conclusion}

This paper develops a bounded stochastic correlation extension of the Vasicek credit risk model by setting
\[
R_t=\cos^2(\varphi_t),
\]
with \(\varphi_t\) evolving as either Circular Brownian motion or a von Mises process. The construction enforces \(R_t\in[0,1]\), preserves positive semidefiniteness for portfolios of arbitrary size, and allows correlation to vary over time without sacrificing the one-factor structure.

A key feature of the framework is the separation between terminal and path-dependent credit quantities. Conditional terminal asset distributions depend on the correlation path only through
\[
\overline R_T=\frac{1}{T}\int_0^T R_s\,ds,
\]
so terminal loss distributions and joint-default probabilities can be evaluated by averaging standard expressions over the law of \(\overline R_T\). Joint survival, first-to-default, and joint first-passage probabilities instead depend on the ordered asset and correlation paths and are evaluated by direct simulation.

The simulations show that terminal joint default and joint first passage need not coincide and that higher initial correlation shifts probability toward concordant outcomes. It can raise joint default and joint first passage while also increasing joint survival and reducing first-to-default. The experiments also distinguish the effects of angular volatility, mean reversion, and stationary concentration on finite-horizon credit risk.

Future work could extend the framework to multivariate and sector-specific correlation structures. It may also be useful for multi-name credit products such as first-to-default and \(k\)th-to-default swaps, basket derivatives, CDO tranches, and counterparty-credit applications.

\section*{Declaration on AI usage}

The proof of the result in Section 3.2 was obtained after conversations with ChatGPT. Codex was used to implement the models and the estimation alogrithm presented in the paper. The authors have verified the details and remain responsible for the output.

\printbibliography
\end{document}